\documentclass[11pt]{article}
\usepackage[a4paper,margin=1in]{geometry}
\usepackage{amsmath,amssymb,amsthm,mathtools}
\usepackage{graphicx}
\graphicspath{{./}{output/}}
\usepackage{enumitem}
\usepackage{hyperref}
\usepackage[nameinlink,noabbrev]{cleveref}
\usepackage{booktabs}
\usepackage{microtype}

\theoremstyle{plain}
\newtheorem{theorem}{Theorem}[section]
\newtheorem{proposition}[theorem]{Proposition}
\newtheorem{corollary}[theorem]{Corollary}
\newtheorem{lemma}[theorem]{Lemma}
\theoremstyle{definition}
\newtheorem{definition}[theorem]{Definition}
\newtheorem{assumption}[theorem]{Assumption}
\newtheorem{remark}[theorem]{Remark}

\title{Entropy-Rate Selection for Partially Observed Processes}
\author{Oleg Kiriukhin\\City University of Hong Kong\\\texttt{okiriukh@cityu.edu.hk}}
\date{April 2026}

\begin{document}
\maketitle
\begin{abstract}

I formulate an entropy-rate maximization problem at the observable level for stochastic processes observed through an information-reducing observation map. For a visible stationary law, the map determines an observational fiber of hidden stationary laws generating that law. In the finite-state finite-memory setting, retained visible constraints determine a feasible class of stationary $(r+1)$-block laws, and the entropy maximizer is defined as the entropy-rate maximizer on this class.

The paper formulates entropy-rate maximization on feasible classes induced by partial observability and develops a structural theory for the resulting maximizer. I prove existence and uniqueness of the maximizer, with uniqueness under a fixed-context-marginal hypothesis and, more generally, via a strict-concavity characterization by row proportionality. Two global characterization regimes are central: a fixed one-point marginal yields the i.i.d.\ maximizer, and a fixed $r$-block law yields the $(r-1)$-step Markov extension. The gap functional equals a conditional mutual information and vanishes exactly at the maximizing completion. I also derive optimality conditions, local geometry of the maximizer, a latent random-mapping realization that leaves the visible law unchanged, and a local empirical consistency theorem, and illustrate the framework by an aliased hidden-state example.


\end{abstract}

\section{Introduction}

Stochastic models are often underidentified by the information structure through which they are observed. Distinct hidden mechanisms may generate the same visible law, so one faces an observational equivalence class rather than a uniquely recoverable latent model. I take the observation experiment, together with the retained visible observables, as primitive and ask whether they determine a preferred visible completion within a finite-dimensional block-law class. This perspective is motivated by Blackwell's comparison of experiments, where the primitive object is the information structure rather than an externally chosen parametric family~\cite{blackwell1951,mckinsey1982}.

Entropy maximization under constraints and entropy calculations for hidden processes are classical themes. The novelty here lies in formulating an entropy-rate maximization problem at the observable level on feasible classes induced by partial observability, and in developing the corresponding structural theory of the maximizer. The emphasis is therefore on the visible law determined by the experiment rather than on selecting a hidden model from an observationally equivalent family.

For a visible law, I define the observational fiber as the class of hidden stationary laws generating that law under a fixed observation experiment. Restricting to finite alphabets and finite memory, I work with stationary $(r+1)$-block laws on the visible alphabet. In block-law coordinates, stationarity is a finite-dimensional linear constraint and the feasible set is a convex subset of the simplex on $A^{r+1}$. An entropy maximizer is an entropy-rate maximizer on the resulting feasible class, and it selects the visible completion with maximal residual uncertainty, hence minimal serial organization beyond what the retained observables force.

Two global characterization regimes are central. If the retained observable fixes the one-point marginal, the maximizer is the i.i.d.\ process with that marginal. If it fixes the full stationary $r$-block law, the maximizer is the $(r-1)$-step Markov extension. In the latter regime, the gap functional is a conditional mutual information and vanishes exactly at the maximizing completion.

The main results are as follows. I define observational fibers and the corresponding feasible classes determined by the retained observables at the visible block-law level, and prove existence and uniqueness of the entropy maximizer, with uniqueness under a fixed-context-marginal hypothesis and, more generally, via a strict-concavity characterization by row proportionality. I then prove two global characterization theorems: fixed one-point marginals yield the i.i.d.\ maximizer, and fixed $r$-block laws yield the $(r-1)$-step Markov maximizer. I also identify the gap functional with conditional mutual information, which vanishes exactly at the maximizing completion.

At the local level, I derive optimality conditions and a kernel formula, identify the intrinsic tangent space and restricted Hessian on a fixed-support face, prove a strict-concavity criterion on affine slices of a common positive-support face, and obtain local geometry together with a quadratic expansion of the gap functional. I also prove that any selected visible law admits a latent random-mapping realization that leaves the visible law unchanged, establish a local empirical consistency theorem under a full-rank moment-map hypothesis, and include an aliased hidden-state example showing that a maximizing visible completion need not resolve hidden underidentification.

\section{Observational Fibers and Entropy Maximization}\label{sec:observational-fibers}

Let $H$ and $A$ be finite alphabets, and let
\begin{equation}\label{eq:observation-map}
\Pi:\mathcal P_{\mathrm{stat}}(H^{\mathbb Z})\to \mathcal P_{\mathrm{stat}}(A^{\mathbb Z}),
\qquad
Q\mapsto \Pi_\#Q
\end{equation}
be a fixed measurable observation map on stationary path laws. For a visible stationary law
\begin{equation*}
\nu\in\mathcal P_{\mathrm{stat}}(A^{\mathbb Z}),
\end{equation*}
define the observational fiber
\begin{equation}\label{eq:observational-fiber}
\mathcal E_\Pi(\nu):=
\{Q\in\mathcal P_{\mathrm{stat}}(H^{\mathbb Z}) : \Pi_\#Q=\nu\}.
\end{equation}

To obtain a finite-dimensional selector, fix a memory length $r\ge 1$ and work with stationary $(r+1)$-block laws on the visible alphabet. Let
\begin{equation*}u(c,a),\qquad c\in A^r,\ a\in A,
\end{equation*}
denote a probability distribution on $A^{r+1}$.

\begin{definition}[Context marginal]
For a block law $u$ on $A^{r+1}$, define its context marginal by
\begin{equation*}
\eta_u(c):=\sum_{a\in A}u(c,a),\qquad c\in A^r.
\end{equation*}
\end{definition}

\begin{definition}[Stationary consistency]
A block law $u$ on $A^{r+1}$ is called stationary-consistent if its left and right $r$-block marginals agree, that is,
\begin{equation*}
\sum_{\alpha\in A}u(\alpha,c_1,\dots,c_r)
=
\sum_{\beta\in A}u(c_1,\dots,c_r,\beta)
\end{equation*}
for every $(c_1,\dots,c_r)\in A^r$.
\end{definition}

Fix retained observable features
\begin{equation*}
G_1,\dots,G_m:A^{r+1}\to\mathbb R,
\end{equation*}
and, for each visible stationary law $\nu$, define the associated target vector
\begin{equation*}
b(\nu):=\bigl(b_1(\nu),\dots,b_m(\nu)\bigr)\in\mathbb R^m.
\end{equation*}
In the finite-state setting, the induced visible feasible class is
\begin{equation}\label{eq:feasible-class}
\mathcal U_\Pi(\nu):=
\Bigl\{u\in\Delta(A^{r+1}) : u \text{ is stationary-consistent and }
\sum_{c,a}u(c,a)G_j(c,a)=b_j(\nu),\ j=1,\dots,m
\Bigr\},
\end{equation}
Whenever $\eta_u(c)>0$ on the active support, define the induced conditional kernel by
\begin{equation}\label{eq:conditional-kernel}
p_u(a\mid c):=\frac{u(c,a)}{\eta_u(c)}.
\end{equation}
The entropy-rate functional is then
\begin{equation}\label{eq:entropy-rate}
J(u):=-\sum_{c\in A^r}\sum_{a\in A}u(c,a)\log\frac{u(c,a)}{\eta_u(c)}.
\end{equation}
Equivalently,
\begin{equation}\label{eq:entropy-rate-alt}
J(u)=\sum_{c\in A^r}\eta_u(c)\,H\bigl(p_u(\cdot\mid c)\bigr).
\end{equation}
For stationary finite-state Markov processes, this is the entropy rate
\begin{equation*}
H(X_r\mid X_0^{r-1}).
\end{equation*}

\begin{definition}[Entropy maximizer]
An entropy maximizer is any maximizer
\begin{equation*}u^\star\in\arg\max_{u\in\mathcal U_\Pi(\nu)}J(u).
\end{equation*}
Whenever $\eta_{u^\star}(c)>0$, the induced active-support kernel is
\begin{equation*}
p^\star(a\mid c):=\frac{u^\star(c,a)}{\eta_{u^\star}(c)}.
\end{equation*}
\end{definition}

\begin{assumption}[Global standing assumptions]
\leavevmode
\begin{enumerate}[label=(A\arabic*),leftmargin=2.5em]
\item The retained observable constraints are finitely many linear equalities in the block-law coordinates $u(c,a)$.
\item The feasible set $\mathcal U_\Pi(\nu)$ is nonempty.
\item For the global uniqueness theorem, the feasible set fixes the context marginal: there exists $\bar\eta$ such that $\eta_u=\bar\eta$ for all $u\in\mathcal U_\Pi(\nu)$.
\end{enumerate}
\end{assumption}

\begin{assumption}[Local regularity assumptions]
\leavevmode
\begin{enumerate}[label=(L\arabic*),leftmargin=2.5em]
\item The selected point $u^\star$ lies in the relative interior of a fixed-support face.
\item A smooth local chart $(b,\xi)\mapsto u(b,\xi)$ exists near $u^\star$, where $b$ are retained observable coordinates and $\xi$ are fiber coordinates.
\end{enumerate}
\end{assumption}

On the positive-support region, block-law and kernel parametrizations are equivalent:
\begin{equation}\label{eq:block-kernel}u(c,a)=\eta_u(c)p_u(a\mid c).
\end{equation}

\section{Existence, Uniqueness, and Optimality Conditions}

\begin{theorem}[Existence]\label{thm:existence}
Under Assumptions \textnormal{(A1)--(A2)}, the optimization problem
\begin{equation*}
\max_{u\in\mathcal U_\Pi(\nu)}J(u)
\end{equation*}
admits at least one maximizer.
\end{theorem}

\begin{proof}
The variable $u$ lies in the probability simplex on the finite set $A^{r+1}$. The
stationary-consistency equations and the experiment-induced restrictions are linear in $u(c,a)$, so
the feasible set $\mathcal U_\Pi(\nu)$ is a closed subset of a finite-dimensional simplex. Hence
$\mathcal U_\Pi(\nu)$ is compact and convex.

The map $x\mapsto -x\log x$ extends continuously to $[0,1]$ by the convention $0\log 0=0$.
Therefore the function
\begin{equation*}
u\mapsto -\sum_{c,a}u(c,a)\log u(c,a)
\end{equation*}
is continuous on the simplex. Since
\begin{equation*}
J(u)= -\sum_{c,a}u(c,a)\log u(c,a)+\sum_c \eta_u(c)\log\eta_u(c),
\end{equation*}
and the marginal map $u\mapsto \eta_u$ is linear, $J$ is continuous on
$\mathcal U_\Pi(\nu)$. By compactness, $J$ attains its maximum on $\mathcal U_\Pi(\nu)$.
\end{proof}

\begin{theorem}[Uniqueness under fixed context marginals]\label{thm:uniqueness}
Assume the hypotheses of \cref{thm:existence}, and assume in addition that there exists a probability
vector $\bar\eta$ on $A^r$ such that $\eta_u=\bar\eta$ for every
$u\in\mathcal U_\Pi(\nu)$. Then the maximizer of $J$ over $\mathcal U_\Pi(\nu)$ is unique.
\end{theorem}

\begin{proof}
Assumption (A3) states that the feasible set fixes the context marginal: there exists $\bar\eta$
such that $\eta_u=\bar\eta$ for every feasible $u$. Hence on $\mathcal U_\Pi(\nu)$ one has
\begin{equation*}
J(u)= -\sum_{c,a}u(c,a)\log u(c,a)+\sum_c \bar\eta(c)\log\bar\eta(c).
\end{equation*}
The second term is constant over the feasible set, so maximizing $J$ is equivalent to maximizing
the Shannon entropy of the block law itself:
\begin{equation*}
H(u):=-\sum_{c,a}u(c,a)\log u(c,a).
\end{equation*}
Shannon entropy is strictly concave on a simplex~\cite{cover-thomas2006}. Therefore $H(u)$ is strictly concave on the
affine hull of the feasible set, and hence on the convex set $\mathcal U_\Pi(\nu)$. A strictly
concave function has at most one maximizer on a convex set. Existence was proved in
\cref{thm:existence}, so the maximizer is unique.
\end{proof}

\begin{proposition}[Concavity and equality criterion]\label{prop:concavity-equality}
Let $u$ and $v$ be stationary $(r+1)$-block laws on $A^{r+1}$, and let $t\in(0,1)$. Then
\begin{equation*}
J(tu+(1-t)v)\ge tJ(u)+(1-t)J(v).
\end{equation*}
Moreover, equality holds if and only if for every context $c\in A^r$ the rows $u(c,\cdot)$ and $v(c,\cdot)$ are proportional. Equivalently, equality holds if and only if
\begin{equation*}
\eta_v(c)u(c,a)=\eta_u(c)v(c,a)
\qquad\text{for every }c\in A^r,\ a\in A.
\end{equation*}
On any common positive-support face, this is further equivalent to equality of the conditional kernels in every active context.
\end{proposition}

\begin{proof}
Write
\begin{equation*}
w:=tu+(1-t)v.
\end{equation*}
For each context $c\in A^r$, define the row masses
\begin{equation*}
\eta_u(c)=\sum_a u(c,a),
\qquad
\eta_v(c)=\sum_a v(c,a),
\qquad
\eta_w(c)=t\eta_u(c)+(1-t)\eta_v(c).
\end{equation*}
Then
\begin{equation*}
J(u)=\sum_c \eta_u(c)H\bigl(p_u(\cdot\mid c)\bigr),
\qquad
J(v)=\sum_c \eta_v(c)H\bigl(p_v(\cdot\mid c)\bigr),
\qquad
J(w)=\sum_c \eta_w(c)H\bigl(p_w(\cdot\mid c)\bigr),
\end{equation*}
where the row conditional laws are defined whenever the corresponding row mass is positive. For a fixed context $c$ with $\eta_w(c)>0$, define
\begin{equation*}
\alpha_c:=\frac{t\eta_u(c)}{\eta_w(c)},
\qquad
1-\alpha_c:=\frac{(1-t)\eta_v(c)}{\eta_w(c)}.
\end{equation*}
Then $\alpha_c\in[0,1]$ and, for each symbol $a\in A$,
\begin{equation*}
p_w(a\mid c)=\alpha_c p_u(a\mid c)+(1-\alpha_c)p_v(a\mid c)
\end{equation*}
whenever both row masses are positive. By concavity of Shannon entropy on the simplex~\cite{cover-thomas2006},
\begin{equation*}
H\bigl(p_w(\cdot\mid c)\bigr)
\ge
\alpha_c H\bigl(p_u(\cdot\mid c)\bigr)+(1-\alpha_c)H\bigl(p_v(\cdot\mid c)\bigr).
\end{equation*}
Multiplying by $\eta_w(c)$ yields
\begin{equation*}
\eta_w(c)H\bigl(p_w(\cdot\mid c)\bigr)
\ge
t\eta_u(c)H\bigl(p_u(\cdot\mid c)\bigr)
+
(1-t)\eta_v(c)H\bigl(p_v(\cdot\mid c)\bigr).
\end{equation*}
Summing over $c\in A^r$ gives the asserted concavity inequality.

For the equality statement, equality in the entropy concavity inequality for a fixed context $c$ holds if and only if the two conditional distributions in that row coincide whenever both mixture weights are positive. Equivalently, whenever $\eta_u(c),\eta_v(c)>0$,
\begin{equation*}
p_u(a\mid c)=p_v(a\mid c)
\qquad\text{for every }a\in A.
\end{equation*}
Multiplying through by the row masses gives
\begin{equation*}
\eta_v(c)u(c,a)=\eta_u(c)v(c,a)
\qquad\text{for every }a\in A.
\end{equation*}
If one of the row masses is zero, equality can hold only when the other row is also zero, which is again exactly the same proportionality condition. Therefore equality in the global concavity inequality holds if and only if the displayed proportionality identities hold in every context. These identities are equivalent to rowwise proportionality of $u(c,\cdot)$ and $v(c,\cdot)$.

Finally, on a common positive-support face all active row masses are positive, so dividing by $\eta_u(c)\eta_v(c)$ shows that the same condition is equivalent to
\begin{equation*}
p_u(\cdot\mid c)=p_v(\cdot\mid c)
\qquad\text{for every active context }c.
\end{equation*}

\end{proof}

\begin{theorem}[Strict-concavity characterization by row proportionality]\label{thm:uniqueness-beyond-a3}
Let $K$ be a convex set of stationary $(r+1)$-block laws on $A^{r+1}$. Then the following are equivalent:
\begin{enumerate}[label=(\roman*),leftmargin=2.5em]
\item $J$ is strictly concave on $K$.
\item For every distinct pair $u,v\in K$, there exist a context $c\in A^r$ and a symbol $a\in A$ such that
\begin{equation*}
\eta_v(c)u(c,a)\ne \eta_u(c)v(c,a).
\end{equation*}
\item No two distinct points of $K$ are rowwise proportional in every context.
\end{enumerate}
If, in addition, every element of $K$ belongs to a common positive-support face, then these conditions are also equivalent to:
\begin{enumerate}[label=(\roman*),leftmargin=2.5em,start=4]
\item For every distinct pair $u,v\in K$, there exists an active context $c$ such that
\begin{equation*}
p_u(\cdot\mid c)\ne p_v(\cdot\mid c).
\end{equation*}
\end{enumerate}
\end{theorem}

\begin{proof}
The equivalence of \textnormal{(ii)} and \textnormal{(iii)} is the coordinate form of rowwise proportionality. If all elements of $K$ belong to a common positive-support face, then for each active context $c$ every row mass is positive, and dividing the identity
\begin{equation*}
\eta_v(c)u(c,a)=\eta_u(c)v(c,a)
\end{equation*}
by $\eta_u(c)\eta_v(c)$ shows that \textnormal{(ii)} is equivalent to \textnormal{(iv)}.

Assume \textnormal{(i)}. If \textnormal{(ii)} failed, then there would exist distinct $u,v\in K$ such that
\begin{equation*}
\eta_v(c)u(c,a)=\eta_u(c)v(c,a)
\qquad\text{for every }c\in A^r,\ a\in A.
\end{equation*}
By \cref{prop:concavity-equality}, equality would then hold along the whole segment joining $u$ and $v$:
\begin{equation*}
J((1-t)u+tv)=(1-t)J(u)+tJ(v)
\qquad\text{for every }t\in(0,1),
\end{equation*}
which contradicts strict concavity on $K$. Thus \textnormal{(i)} implies \textnormal{(ii)}.

Assume \textnormal{(ii)}. Let $u,v\in K$ be distinct and let $t\in(0,1)$. Since $K$ is convex, the point
\begin{equation*}
w:=(1-t)u+tv
\end{equation*}
belongs to $K$. By \cref{prop:concavity-equality},
\begin{equation*}
J(w)\ge (1-t)J(u)+tJ(v),
\end{equation*}
and equality holds if and only if $u$ and $v$ are rowwise proportional in every context. Hypothesis \textnormal{(ii)} rules out that equality case for distinct $u$ and $v$. Therefore
\begin{equation*}
J((1-t)u+tv)>(1-t)J(u)+tJ(v)
\qquad\text{for all distinct }u,v\in K,\ t\in(0,1),
\end{equation*}
which is exactly strict concavity of $J$ on $K$. Thus \textnormal{(ii)} implies \textnormal{(i)}.

This proves the equivalence of \textnormal{(i)}--\textnormal{(iii)}, and also of \textnormal{(i)}--\textnormal{(iv)} on a common positive-support face.
\end{proof}

\begin{corollary}[Global uniqueness on the feasible class beyond \textnormal{(A3)}]\label{cor:global-uniqueness-beyond-a3}
Assume \textnormal{(A1)--(A2)}. Suppose that for every distinct feasible pair $u,v\in \mathcal U_\Pi(\nu)$ there exist a context $c\in A^r$ and a symbol $a\in A$ such that
\begin{equation*}
\eta_v(c)u(c,a)\ne \eta_u(c)v(c,a).
\end{equation*}
Then the maximizer of $J$ over $\mathcal U_\Pi(\nu)$ is unique.
\end{corollary}

\begin{proof}
By \cref{thm:existence}, the feasible class $\mathcal U_\Pi(\nu)$ is nonempty, compact, and convex, and $J$ attains a maximum on it. The stated hypothesis implies, via \cref{thm:uniqueness-beyond-a3}, that $J$ is strictly concave on $\mathcal U_\Pi(\nu)$. A strictly concave function has at most one maximizer on a convex set, so the maximizer is unique.
\end{proof}

\begin{corollary}[Kernel-separation criterion on a common positive-support face]
Assume \textnormal{(A1)--(A2)} and suppose every feasible block law in $\mathcal U_\Pi(\nu)$ belongs to a common positive-support face. Assume, in addition, that for every distinct feasible pair $u,v\in \mathcal U_\Pi(\nu)$ there exists an active context $c$ such that
\begin{equation*}
p_u(\cdot\mid c)\ne p_v(\cdot\mid c).
\end{equation*}
Then the maximizer of $J$ over $\mathcal U_\Pi(\nu)$ is unique.
\end{corollary}

\begin{proof}
Under the common positive-support assumption, \cref{thm:uniqueness-beyond-a3} shows that the kernel-separation hypothesis is equivalent to strict concavity of $J$ on $\mathcal U_\Pi(\nu)$. The conclusion then follows from \cref{cor:global-uniqueness-beyond-a3}.
\end{proof}

\begin{remark}
\cref{thm:uniqueness-beyond-a3} is strictly broader than \cref{thm:uniqueness}. Assumption \textnormal{(A3)} eliminates all variation of the context marginal across the feasible set, whereas the broader characterization excludes only the flat directions along which every row changes by a scalar factor while the conditional kernels remain unchanged.
\end{remark}

\paragraph{A sufficient condition for Assumption \textnormal{(A3)}.}
Suppose that for each context $c\in A^r$, the context-indicator function
\begin{equation*}
I_c(c',a):=\mathbf 1\{c'=c\},
\qquad (c',a)\in A^{r+1},
\end{equation*}
belongs to the linear span of the retained observable features $G_1,\dots,G_m$ together with the constant function $1$ on $A^{r+1}$. Then the feasible set $\mathcal U_\Pi(\nu)$ fixes the full context marginal, so Assumption \textnormal{(A3)} holds and \cref{thm:uniqueness} applies.

\begin{proof}
For any feasible block law $u\in \mathcal U_\Pi(\nu)$,
\begin{equation*}
\eta_u(c)=\sum_{a\in A}u(c,a)=\sum_{c',a}u(c',a)I_c(c',a).
\end{equation*}
By hypothesis, one may write
\begin{equation*}
I_c=\alpha_{c,0}+\sum_{j=1}^m \alpha_{c,j}G_j
\end{equation*}
for suitable coefficients depending on $c$. Therefore
\begin{equation*}
\eta_u(c)=\alpha_{c,0}\sum_{c',a}u(c',a)+\sum_{j=1}^m \alpha_{c,j}\sum_{c',a}u(c',a)G_j(c',a).
\end{equation*}
The first term equals $\alpha_{c,0}$ by normalization, and each remaining term is fixed on $\mathcal U_\Pi(\nu)$ by the retained moment constraints. Hence $\eta_u(c)$ is the same for every feasible $u$. Hence the full context marginal is fixed on $\mathcal U_\Pi(\nu)$.
\end{proof}

\section{Global Characterization Theorems}

\begin{theorem}[Fixed marginal case]
\label{thm:fixed-marginal-iid}
Let $A$ be a finite alphabet, and let $\pi=(\pi_a)_{a\in A}$ be a probability vector on $A$.
Consider the feasible class
\begin{equation*}
U(\pi)
:=
\left\{
 u=(u_{a,b})_{a,b\in A} :
 u_{a,b}\ge 0,
 \sum_{b\in A}u_{a,b}=\pi_a,
 \sum_{a\in A}u_{a,b}=\pi_b
\right\}.
\end{equation*}
Then the entropy-rate functional
\begin{equation*}
J(u)
=
-\sum_{a,b\in A}u_{a,b}\log\frac{u_{a,b}}{\pi_a}
\end{equation*}
has the unique maximizer
\begin{equation*}u^\star_{a,b}=\pi_a\pi_b,
\qquad a,b\in A.
\end{equation*}
That is, the entropy maximizer is the i.i.d.\ process with one-point law $\pi$.
\end{theorem}

\begin{proof}
Let $(X_t)_{t\in\mathbb{Z}}$ be a stationary first-order process with two-block law $u$.
Because $u\in U(\pi)$, both $X_0$ and $X_1$ have marginal law $\pi$.
By the entropy-rate definition,
\begin{equation*}
J(u)=H(X_1\mid X_0).
\end{equation*}
By monotonicity of conditional entropy, applied to the conditioning variables $X_0$ and the trivial sigma-field~\cite{cover-thomas2006},
\begin{equation*}
H(X_1\mid X_0)\le H(X_1)=H(\pi).
\end{equation*}
Hence $J(u)\le H(\pi)$ for every $u\in U(\pi)$.

Now define $u^\star$ by $u^\star_{a,b}=\pi_a\pi_b$.
Then $u^\star\in U(\pi)$ because
\begin{equation*}
\sum_{b\in A}u^\star_{a,b}=\pi_a\sum_{b\in A}\pi_b=\pi_a,
\qquad
\sum_{a\in A}u^\star_{a,b}=\pi_b\sum_{a\in A}\pi_a=\pi_b.
\end{equation*}
The corresponding process is i.i.d.\ with one-point law $\pi$.
Therefore $X_1$ is independent of $X_0$, and thus
\begin{equation*}
J(u^\star)=H(X_1\mid X_0)=H(X_1)=H(\pi).
\end{equation*}
Hence $u^\star$ attains the maximum.

Suppose that $u\in U(\pi)$ also satisfies $J(u)=H(\pi)$.
Then
\begin{equation*}
H(X_1\mid X_0)=H(X_1).
\end{equation*}
Equality holds if and only if $X_1$ is independent of $X_0$.
Since both marginals equal $\pi$, independence implies
\begin{equation*}u_{a,b}=\mathbb{P}(X_0=a,X_1=b)=\pi_a\pi_b,
\qquad a,b\in A.
\end{equation*}
Hence $u=u^\star$.
\end{proof}

\begin{corollary}[Binary fixed-mean case]
\label{cor:binary-fixed-mean-iid}
Let $A=\{0,1\}$, and let the retained observable be the stationary mean
\begin{equation*}
m=\mathbb{P}(X_t=1)\in[0,1].
\end{equation*}
Then the unique entropy-rate maximizer among stationary binary first-order laws with mean $m$
is the i.i.d.\ Bernoulli$(m)$ law.
\end{corollary}

\begin{proof}
Apply \cref{thm:fixed-marginal-iid} with $\pi=(1-m,m)$.
\end{proof}

\begin{theorem}[Fixed $r$-block case]
\label{thm:fixed-r-block-markov}
Let $r\ge 1$, let $A$ be a finite alphabet, and let $\mu$ be a probability law on $A^r$.
Adopt the convention that $A^0=\{\varnothing\}$ and that the empty-block marginal equals $\mu(\varnothing)=1$ when $r=1$.
Consider the feasible class
\begin{equation*}
U(\mu)
:=
\left\{
 u=(u(c,a))_{c\in A^r,\ a\in A} :
 \begin{array}{l}
 u(c,a)\ge 0,\\[0.3em]
 \sum_{a\in A}u(c,a)=\mu(c)\quad \text{for every } c\in A^r,\\[0.3em]
 u \text{ is stationary-consistent}
 \end{array}
\right\}.
\end{equation*}
For $u\in U(\mu)$, write
\begin{equation*}
J(u)= -\sum_{c\in A^r}\sum_{a\in A}u(c,a)\log\frac{u(c,a)}{\mu(c)}
 = H(X_r\mid X_0^{r-1}),
\end{equation*}
where $(X_t)$ is any stationary process with $(r+1)$-block law $u$.
When $r=1$, the expression $X_1^{r-1}$ below is interpreted as the trivial conditioning sigma-field, so
\begin{equation*}
H_\mu(X_r\mid X_1^{r-1})=H_\mu(X_1).
\end{equation*}

For $s=(s_1,\dots,s_{r-1})\in A^{r-1}$ with $\mu(s)>0$, define
\begin{equation*}
q(a\mid s):=\frac{\mu(s,a)}{\mu(s)},
\qquad a\in A.
\end{equation*}
For $s\in A^{r-1}$ with $\mu(s)=0$, fix an arbitrary probability vector $q(\cdot\mid s)$ on $A$.
Define
\begin{equation*}u^\star(c,a)=\mu(c)\,q(a\mid c_2,\dots,c_r),
\qquad c=(c_1,\dots,c_r)\in A^r,\ a\in A.
\end{equation*}
Then $u^\star\in U(\mu)$ and $u^\star$ is the unique maximizer of $J$ on $U(\mu)$.
That is, among all stationary $(r+1)$-block laws extending $\mu$, the entropy-rate selector is the $(r-1)$-step Markov extension.
\end{theorem}

\begin{proof}
Let $u\in U(\mu)$, and let $(X_t)_{t\in\mathbb{Z}}$ be a stationary process with $(r+1)$-block law $u$.
Because $\sum_{a\in A}u(c,a)=\mu(c)$ for every $c\in A^r$, the law of $(X_0,\dots,X_{r-1})$ is $\mu$.
Because $u$ is stationary-consistent, the law of $(X_1,\dots,X_r)$ is also $\mu$.
By definition,
\begin{equation*}
J(u)=H(X_r\mid X_0^{r-1}).
\end{equation*}
Since $X_1^{r-1}$ is a function of $X_0^{r-1}$, conditioning on the larger sigma-field cannot increase entropy~\cite{cover-thomas2006}, so
\begin{equation*}
H(X_r\mid X_0^{r-1})\le H(X_r\mid X_1^{r-1}).
\end{equation*}
Moreover, the joint law of $(X_1,\dots,X_r)$ is fixed and equal to $\mu$, so the quantity on the right-hand side depends only on $\mu$.
Therefore every feasible law satisfies
\begin{equation*}
J(u)\le H_{\mu}(X_r\mid X_1^{r-1}).
\end{equation*}

To verify $u^\star\in U(\mu)$, note first that
\begin{equation*}
\sum_{a\in A}u^\star(c,a)=\mu(c)\sum_{a\in A}q(a\mid c_2,\dots,c_r)=\mu(c)
\qquad c\in A^r,
\end{equation*}
so $u^\star$ has the required context marginal.
It remains to verify stationarity-consistency.
Fix $d=(d_1,\dots,d_r)\in A^r$.
Then
\begin{equation*}
\sum_{a\in A}u^\star(a,d)
=
\sum_{a\in A}\mu(a,d_1,\dots,d_{r-1})\,q(d_r\mid d_1,\dots,d_{r-1}).
\end{equation*}
If $\mu(d_1,\dots,d_{r-1})>0$, the factor $q(d_r\mid d_1,\dots,d_{r-1})$ is equal to
\begin{equation*}
\frac{\mu(d_1,\dots,d_r)}{\mu(d_1,\dots,d_{r-1})}.
\end{equation*}
Hence
\begin{equation*}
\sum_{a\in A}u^\star(a,d)
=
\Bigl(\sum_{a\in A}\mu(a,d_1,\dots,d_{r-1})\Bigr)
\frac{\mu(d_1,\dots,d_r)}{\mu(d_1,\dots,d_{r-1})}.
\end{equation*}
By marginalization of the probability law $\mu$,
\begin{equation*}
\sum_{a\in A}\mu(a,d_1,\dots,d_{r-1})=\mu(d_1,\dots,d_{r-1}).
\end{equation*}
Therefore
\begin{equation*}
\sum_{a\in A}u^\star(a,d)=\mu(d_1,\dots,d_r).
\end{equation*}
If instead $\mu(d_1,\dots,d_{r-1})=0$, then every term $\mu(a,d_1,\dots,d_{r-1})$ is zero, so
\begin{equation*}
\sum_{a\in A}u^\star(a,d)=0.
\end{equation*}
Also $\mu(d_1,\dots,d_r)\le \mu(d_1,\dots,d_{r-1})=0$, hence $\mu(d_1,\dots,d_r)=0$.
Therefore in all cases
\begin{equation*}
\sum_{a\in A}u^\star(a,d)=\mu(d)=\sum_{a\in A}u^\star(d,a),
\end{equation*}
which is stationarity-consistency.
Thus $u^\star\in U(\mu)$.

Fix a context $c=(c_1,\dots,c_r)$ with $\mu(c)>0$.
For every $a\in A$, the definition of $u^\star$ and the identity $\sum_{b\in A}u^\star(c,b)=\mu(c)$ give
\begin{equation*}
\mathbb{P}_{u^\star}(X_r=a\mid X_0^{r-1}=c)
=
\frac{u^\star(c,a)}{\mu(c)}
=
q(a\mid c_2,\dots,c_r).
\end{equation*}
The right-hand side depends on $c$ only through $(c_2,\dots,c_r)$, so
\begin{equation*}
\mathbb{P}_{u^\star}(X_r=a\mid X_0^{r-1})=
\mathbb{P}_{u^\star}(X_r=a\mid X_1^{r-1})
\qquad\text{for every } a\in A,
\end{equation*}
that is, $X_r\perp X_0\mid X_1^{r-1}$. Therefore
\begin{equation*}
H(X_r\mid X_0^{r-1})=H(X_r\mid X_1^{r-1}).
\end{equation*}
Since the law of $(X_1,\dots,X_r)$ is $\mu$,
\begin{equation*}
J(u^\star)=H_{\mu}(X_r\mid X_1^{r-1}).
\end{equation*}
Hence $u^\star$ attains the maximum.

Suppose that $u\in U(\mu)$ also attains the same maximal value.
Then equality holds in the conditional-entropy inequality:
\begin{equation*}
H(X_r\mid X_0^{r-1})=H(X_r\mid X_1^{r-1}).
\end{equation*}
Since $X_1^{r-1}$ is a function of $X_0^{r-1}$, the difference between these two entropies is the conditional mutual information of $X_0$ and $X_r$ given $X_1^{r-1}$, namely
\begin{equation*}
H(X_r\mid X_1^{r-1})-H(X_r\mid X_0^{r-1})=I(X_0,X_r\mid X_1^{r-1}).
\end{equation*}
Hence equality holds if and only if
\begin{equation*}
I(X_0,X_r\mid X_1^{r-1})=0,
\end{equation*}
that is, if and only if
\begin{equation*}
X_r\perp X_0\mid X_1^{r-1}.
\end{equation*}

Now fix $c=(c_1,\dots,c_r)\in A^r$ with $\mu(c)>0$ and let $a\in A$.
By the chain rule and the conditional independence just established,
\begin{equation*}u(c,a)
=
\mathbb{P}(X_0^{r-1}=c, X_r=a)
=
\mathbb{P}(X_0^{r-1}=c)\,\mathbb{P}(X_r=a\mid X_0^{r-1}=c)
\end{equation*}
\begin{equation*}
=
\mu(c)\,\mathbb{P}(X_r=a\mid X_1^{r-1}=c_2,\dots,c_r)
=
\mu(c)\,q(a\mid c_2,\dots,c_r)
=u^\star(c,a).
\end{equation*}
Thus $u(c,a)=u^\star(c,a)$ for every $c\in A^r$ with $\mu(c)>0$ and every $a\in A$.
If $\mu(c)=0$, then feasibility gives $\sum_{a\in A}u(c,a)=0$.
Since every coordinate is nonnegative, this implies $u(c,a)=0$ for every $a\in A$.
By the definition of $u^\star$ one also has $u^\star(c,a)=0$ for every $a\in A$.
Therefore $u=u^\star$ on all coordinates. 
\end{proof}

\begin{corollary}[Gap and conditional mutual information]
\label{cor:gap-cmi}
Under the hypotheses of \cref{thm:fixed-r-block-markov}, define
\begin{equation*}
\Delta_\mu(u):=H_\mu(X_r\mid X_1^{r-1})-J(u).
\end{equation*}
Then, for every $u\in U(\mu)$,
\begin{equation*}
\Delta_\mu(u)=I_u(X_0,X_r\mid X_1^{r-1})\ge 0.
\end{equation*}
Moreover,
\begin{equation*}
\Delta_\mu(u)=0
\quad\Longleftrightarrow\quad u=u^\star.
\end{equation*}
\end{corollary}

\begin{proof}
Let $u\in U(\mu)$.
Because the law of $(X_1,\dots,X_r)$ is fixed and equal to $\mu$, one has
\begin{equation*}
H_\mu(X_r\mid X_1^{r-1})=H_u(X_r\mid X_1^{r-1}).
\end{equation*}
Using the definition of $J(u)$ from \cref{thm:fixed-r-block-markov},
\begin{equation*}
\Delta_\mu(u)=H_u(X_r\mid X_1^{r-1})-H_u(X_r\mid X_0^{r-1}).
\end{equation*}
Since $X_0^{r-1}=(X_0,X_1^{r-1})$, the defining identity for conditional mutual information gives
\begin{equation*}
I_u(X_0,X_r\mid X_1^{r-1})
=
H_u(X_r\mid X_1^{r-1})-H_u(X_r\mid X_0,X_1^{r-1})
=
H_u(X_r\mid X_1^{r-1})-H_u(X_r\mid X_0^{r-1}).
\end{equation*}
Therefore
\begin{equation*}
\Delta_\mu(u)=I_u(X_0,X_r\mid X_1^{r-1})\ge 0.
\end{equation*}

Finally, $\Delta_\mu(u)=0$ if and only if $I_u(X_0,X_r\mid X_1^{r-1})=0$.
This is equivalent to
\begin{equation*}
X_r\perp X_0\mid X_1^{r-1}.
\end{equation*}
By the uniqueness statement already proved in \cref{thm:fixed-r-block-markov}, that conditional independence relation is equivalent to $u=u^\star$.
Hence $\Delta_\mu(u)=0$ if and only if $u=u^\star$.
\end{proof}

The experiment-induced linear restrictions are
\begin{equation}\label{eq:moment-constraints}
\sum_{c,a}u(c,a)G_j(c,a)=b_j,
\qquad j=1,\dots,m.
\end{equation}
The stationarity-consistency constraints can be written as
\begin{equation}\label{eq:stationarity-constraints}
\sum_{a\in A}u(c,a)-\sum_{a\in A}u(a,c)=0,
\qquad c\in A^r,
\end{equation}
where $u(a,c)$ denotes the block whose suffix of length $r$ is $c$.

For a block $(c,a)\in A^{r+1}$, let $\sigma(c,a)$ denote its suffix of length $r$.

\begin{proposition}[Optimality conditions]\label{prop:kkt}
Assume the unique maximizer $u^\star$ lies in the relative interior of a fixed-support face
of the feasible set. Then there exist multipliers $\lambda_1,\dots,\lambda_m$, a scalar
$\gamma$, and stationarity multipliers $\psi(c)$, $c\in A^r$, such that for each active
coordinate $(c,a)$,
\begin{equation*}
\log\frac{u^\star(c,a)}{\eta_{u^\star}(c)}
=
-\gamma
-\sum_{j=1}^m \lambda_j G_j(c,a)
-\psi(c)
+\psi(\sigma(c,a)).
\end{equation*}
Equivalently,
\begin{equation*}u^\star(c,a)=\eta_{u^\star}(c)
\exp\!\Bigl(
-\gamma
-\sum_{j=1}^m \lambda_j G_j(c,a)
-\psi(c)
+\psi(\sigma(c,a))
\Bigr).
\end{equation*}
\end{proposition}

\begin{proof}
Restrict the feasible set to the affine slice determined by the linear experiment restrictions, the
normalization constraint, and the stationarity equations, and then to the fixed-support face on
which $u^\star(c,a)>0$. On that relative interior, $J$ is $C^1$, so the usual finite-dimensional
Lagrange-multiplier conditions apply.

Consider the Lagrangian
\begin{equation*}
L(u,\lambda,\psi,\gamma)
=
J(u)
-\sum_{j=1}^m \lambda_j\Bigl(\sum_{c,a}u(c,a)G_j(c,a)-b_j\Bigr)
-\gamma\Bigl(\sum_{c,a}u(c,a)-1\Bigr)
-\sum_{c\in A^r}\psi(c)\Bigl(\sum_a u(c,a)-\sum_a u(a,c)\Bigr).
\end{equation*}
Since
\begin{equation*}
J(u)= -\sum_{c,a}u(c,a)\log u(c,a)+\sum_c \eta_u(c)\log\eta_u(c),
\qquad
\eta_u(c)=\sum_a u(c,a),
\end{equation*}
its derivative with respect to an active coordinate $u(c,a)$ is
\begin{equation*}
\frac{\partial}{\partial u(c,a)}
\Bigl[-\sum_{c',a'}u(c',a')\log u(c',a')\Bigr]
=-(\log u(c,a)+1),
\end{equation*}
while
\begin{equation*}
\frac{\partial}{\partial u(c,a)}
\Bigl[\sum_{c'} \eta_u(c')\log\eta_u(c')\Bigr]
=\log\eta_u(c)+1.
\end{equation*}
Hence
\begin{equation*}
\frac{\partial J}{\partial u(c,a)}=\log\eta_u(c)-\log u(c,a).
\end{equation*}

At the interior maximizer $u^\star$, the first-order condition on the fixed-support face reads
\begin{equation*}
0
=
\log\eta_{u^\star}(c)-\log u^\star(c,a)
-\sum_{j=1}^m \lambda_j G_j(c,a)
-\gamma
-\psi(c)
+\psi(\sigma(c,a)).
\end{equation*}
Rearranging and exponentiating gives the result on the active support.
\end{proof}

\begin{corollary}[Kernel representation on the active support]\label{cor:kernel}
On the active support where $\eta_{u^\star}(c)>0$, the induced kernel
\begin{equation*}
p^\star(a\mid c):=\frac{u^\star(c,a)}{\eta_{u^\star}(c)}
\end{equation*}
satisfies
\begin{equation*}
p^\star(a\mid c)
\propto
\exp\!\Bigl(
-\sum_{j=1}^m \lambda_j G_j(c,a)+\psi(\sigma(c,a))
\Bigr),
\qquad a\in A.
\end{equation*}
Equivalently,
\begin{equation*}
p^\star(a\mid c)=
\frac{\exp\!\Bigl(-\sum_{j=1}^m \lambda_j G_j(c,a)+\psi(\sigma(c,a))\Bigr)}
{\sum_{\beta\in A}\exp\!\Bigl(-\sum_{j=1}^m \lambda_j G_j(c,\beta)+\psi(\sigma(c,\beta))\Bigr)}.
\end{equation*}
If, in addition, the term $\psi(\sigma(c,a))$ is constant in $a$ for each fixed context $c$, then this simplifies to the familiar rowwise exponential-family formula
\begin{equation*}
p^\star(a\mid c)=
\frac{\exp\!\bigl(-\sum_{j=1}^m \lambda_j G_j(c,a)\bigr)}
{\sum_{\beta\in A}\exp\!\bigl(-\sum_{j=1}^m \lambda_j G_j(c,\beta)\bigr)}.
\end{equation*}
\end{corollary}

\begin{remark}
In the general formula, the stationarity multipliers enter through $\psi(\sigma(c,a))$, coupling rows across future contexts. The rowwise exponential-family form appears only in the special case of \cref{cor:kernel}.
\end{remark}

\section{Local Geometry and the Gap Functional}

Throughout this section, assume \textnormal{(L1)--(L2)} and work on a fixed-support face. Let $(b,\xi)\mapsto u(b,\xi)$ denote a smooth local chart near $u^\star=u(b_0,\xi_0)$, where $b$ denotes retained observable coordinates and $\xi$ denotes fiber coordinates. Define
\begin{equation}\label{eq:local-objective}
\widetilde J(b,\xi):=J(u(b,\xi)).
\end{equation}
The selector map is then
\begin{equation}\label{eq:selector-map}
s(b)\in\arg\max_{\xi}\widetilde J(b,\xi).
\end{equation}

\paragraph{Intrinsic tangent space on a fixed face.}
Fix an fixed-support face $\mathfrak F\subseteq \Delta(A^{r+1})$ and let $\mathcal A_{\mathfrak F}$ denote the affine space of stationary-consistent block laws supported on $\mathfrak F$ and satisfying the retained observable constraints. Its translation space is
\begin{equation}\label{eq:tangent-space}
\mathrm{Tan}(\mathcal A_{\mathfrak F})
=
\left\{h\in \mathbb R^{A^{r+1}}:\;
\begin{aligned}
&\operatorname{supp}(h)\subseteq \mathfrak F,\quad
\textstyle\sum_{c,a}h(c,a)=0,\\
&\textstyle\sum_{c,a}h(c,a)G_j(c,a)=0\quad(j=1,\dots,m),\\
&\textstyle\sum_{a\in A}h(c,a)=\sum_{a\in A}h(a,c)\quad(c\in A^r)
\end{aligned}
\right\}.
\end{equation}
In particular, any local fiber coordinate $\xi$ on a fixed face may be chosen along a basis of this tangent space.

\begin{proof}
Inside a fixed face, the feasible set is cut out by the normalization constraint, the retained moment equalities, and the stationarity-consistency equations, all of which are affine-linear in the block-law coordinates. The translation space of that affine slice is therefore exactly the set of perturbations satisfying the corresponding homogeneous linear equations. The support restriction records that one remains on the same fixed-support face.
\end{proof}

\paragraph{Restricted Hessian in block-law coordinates.}
Let $u$ be strictly positive on the fixed face $\mathfrak F$, and let $h\in \mathrm{Tan}(\mathcal A_{\mathfrak F})$. Then the second variation of the entropy-rate functional along $h$ is
\begin{equation}\label{eq:hessian}
D^2J(u)[h,h]
=
-\sum_{c\in A^r}\sum_{a\in A}\frac{h(c,a)^2}{u(c,a)}
+\sum_{c\in A^r}\frac{\bigl(\sum_{a\in A}h(c,a)\bigr)^2}{\eta_u(c)}.
\end{equation}
Consequently $D^2J(u)[h,h]\le 0$ for every such $h$. If, in addition, $\sum_{a\in A}h(c,a)=0$ for every context $c$ and $h\neq 0$, then $D^2J(u)[h,h]<0$.

\begin{proof}
Write
\begin{equation*}
J(u)= -\sum_{c,a}u(c,a)\log u(c,a)+\sum_c \eta_u(c)\log\eta_u(c),
\qquad \eta_u(c)=\sum_a u(c,a).
\end{equation*}
For $u_t:=u+th$, differentiation gives
\begin{equation*}
\frac{d^2}{dt^2}\Big|_{t=0}\Bigl[-\sum_{c,a}u_t(c,a)\log u_t(c,a)\Bigr]
=-\sum_{c,a}\frac{h(c,a)^2}{u(c,a)},
\end{equation*}
while
\begin{equation*}
\frac{d^2}{dt^2}\Big|_{t=0}\Bigl[\sum_c \eta_{u_t}(c)\log\eta_{u_t}(c)\Bigr]
=\sum_c \frac{\bigl(\sum_a h(c,a)\bigr)^2}{\eta_u(c)}.
\end{equation*}
Adding the two contributions yields the displayed formula.

For each fixed context $c$, Cauchy--Schwarz gives
\begin{equation*}
\Bigl(\sum_a h(c,a)\Bigr)^2
\le
\Bigl(\sum_a u(c,a)\Bigr)
\Bigl(\sum_a \frac{h(c,a)^2}{u(c,a)}\Bigr)
=
\eta_u(c)\sum_a \frac{h(c,a)^2}{u(c,a)}.
\end{equation*}
Summing over $c$ proves that $D^2J(u)[h,h]\le 0$. If in addition each row sum $\sum_a h(c,a)$ vanishes, then the positive term disappears and one obtains
\begin{equation*}
D^2J(u)[h,h]= -\sum_{c,a}\frac{h(c,a)^2}{u(c,a)}<0
\end{equation*}
for every nonzero $h$.
\end{proof}

The fiber Hessian in \cref{thm:selector-diff,thm:gap} is the restriction of $D^2J(u)$ to the tangent directions selected by the local coordinates. Concavity is automatic on each fixed face; strict negative definiteness is an additional hypothesis.

\begin{proposition}[Null directions of the restricted Hessian]\label{prop:hessian-null}
Let $u$ be strictly positive on the fixed face $\mathfrak F$, and let $h\in \mathrm{Tan}(\mathcal A_{\mathfrak F})$. Then the following are equivalent:
\begin{enumerate}[label=(\roman*),leftmargin=2.5em]
\item $D^2J(u)[h,h]=0$.
\item For every context $c\in A^r$ there exists a scalar $\alpha(c)\in \mathbb R$ such that
\begin{equation*}
h(c,a)=\alpha(c)u(c,a)
\qquad\text{for all }a\in A.
\end{equation*}
\end{enumerate}
In particular, the null directions of the restricted Hessian are exactly the row-rescaling directions that preserve each conditional law $p_u(\cdot\mid c)$.
\end{proposition}

\begin{proof}
By the restricted-Hessian formula proved above,
\begin{equation*}
D^2J(u)[h,h]
=
-\sum_{c\in A^r}\Biggl[
\sum_{a\in A}\frac{h(c,a)^2}{u(c,a)}
-
\frac{\bigl(\sum_{a\in A}h(c,a)\bigr)^2}{\eta_u(c)}
\Biggr].
\end{equation*}
For each fixed context $c$, set
\begin{equation*}
S_c:=\sum_{a\in A}\frac{h(c,a)^2}{u(c,a)}
-
\frac{\bigl(\sum_{a\in A}h(c,a)\bigr)^2}{\eta_u(c)}.
\end{equation*}
The Cauchy--Schwarz inequality applied to the vectors
\begin{equation*}
\Bigl(\frac{h(c,a)}{\sqrt{u(c,a)}}\Bigr)_{a\in A}
\qquad\text{and}\qquad
\bigl(\sqrt{u(c,a)}\bigr)_{a\in A}
\end{equation*}
gives
\begin{equation*}
\Bigl(\sum_{a\in A}h(c,a)\Bigr)^2
\le
\eta_u(c)\sum_{a\in A}\frac{h(c,a)^2}{u(c,a)},
\end{equation*}
so each $S_c\ge 0$. Therefore
\begin{equation*}
D^2J(u)[h,h]= -\sum_{c\in A^r}S_c\le 0.
\end{equation*}
Hence $D^2J(u)[h,h]=0$ if and only if $S_c=0$ for every context $c$.

Fix a context $c$. Equality in Cauchy--Schwarz holds if and only if the two vectors displayed above are linearly dependent. Since $u(c,a)>0$ on the fixed face, this is equivalent to the existence of a scalar $\alpha(c)$ such that
\begin{equation*}
\frac{h(c,a)}{\sqrt{u(c,a)}}=\alpha(c)\sqrt{u(c,a)}
\qquad\text{for all }a\in A,
\end{equation*}
which is exactly
\begin{equation*}
h(c,a)=\alpha(c)u(c,a)
\qquad\text{for all }a\in A.
\end{equation*}
Thus $S_c=0$ if and only if the row $h(c,\cdot)$ is proportional to the row $u(c,\cdot)$. Since this must hold for every context, statements \textnormal{(i)} and \textnormal{(ii)} are equivalent.

Finally, if $h(c,a)=\alpha(c)u(c,a)$ for all $a$, then for every sufficiently small $t$ such that $u+th$ remains in the face,
\begin{equation*}
p_{u+th}(a\mid c)
=
\frac{u(c,a)+th(c,a)}{\eta_u(c)+t\sum_a h(c,a)}
=
\frac{(1+t\alpha(c))u(c,a)}{(1+t\alpha(c))\eta_u(c)}
=
p_u(a\mid c).
\end{equation*}
Therefore these directions change only the row masses and preserve each conditional law.
\end{proof}

\begin{theorem}[Strict concavity on a fixed face]\label{thm:affine-slice-strict}
Let $K$ be a convex subset of $\mathcal A_{\mathfrak F}\cap \operatorname{relint}(\mathfrak F)$. Assume that for every $u\in K$ and every nonzero direction $h\in \mathrm{Tan}(\mathcal A_{\mathfrak F})$, there exists a context $c\in A^r$ such that the row $h(c,\cdot)$ is not proportional to the row $u(c,\cdot)$. Then $J$ is strictly concave on $K$. In particular, if $K$ is nonempty and compact, then $J$ has a unique maximizer on $K$.
\end{theorem}

\begin{proof}
Let $u,v\in K$ be distinct, and define
\begin{equation*}
h:=v-u.
\end{equation*}
Because $K\subseteq \mathcal A_{\mathfrak F}$ and $\mathcal A_{\mathfrak F}$ is affine with translation space $\mathrm{Tan}(\mathcal A_{\mathfrak F})$, one has
\begin{equation*}
h\in \mathrm{Tan}(\mathcal A_{\mathfrak F})
\qquad\text{and}\qquad h\ne 0.
\end{equation*}
For each $t\in(0,1)$, set
\begin{equation*}
\nu_t:=(1-t)u+tv.
\end{equation*}
Since $K$ is convex, $\nu_t\in K\subseteq \mathcal A_{\mathfrak F}\cap \operatorname{relint}(\mathfrak F)$ for every $t\in(0,1)$. Consider the one-variable function
\begin{equation*}
\varphi(t):=J(\nu_t).
\end{equation*}
Because $J$ is $C^2$ on the relative interior of the fixed face, $\varphi$ is twice continuously differentiable on $(0,1)$ and
\begin{equation*}
\varphi''(t)=D^2J(\nu_t)[h,h].
\end{equation*}
By hypothesis, for each $t\in(0,1)$ the nonzero vector $h$ is not rowwise proportional to $\nu_t$ in every context. Therefore \cref{prop:hessian-null} yields
\begin{equation*}
D^2J(\nu_t)[h,h]<0
\qquad\text{for every }t\in(0,1).
\end{equation*}
Hence $\varphi$ is strictly concave on $(0,1)$. In particular,
\begin{equation*}
\varphi(t)>(1-t)\varphi(0)+t\varphi(1)
\qquad\text{for every }t\in(0,1),
\end{equation*}
that is,
\begin{equation*}
J((1-t)u+tv)>(1-t)J(u)+tJ(v)
\qquad\text{for all distinct }u,v\in K,\ t\in(0,1).
\end{equation*}
Thus $J$ is strictly concave on $K$.

If $K$ is nonempty and compact, continuity of $J$ on $K$ implies that $J$ attains its maximum on $K$. Strict concavity on a convex set gives uniqueness of that maximizer.
\end{proof}

\begin{remark}
\cref{thm:affine-slice-strict} is the fixed-face analogue of \cref{thm:uniqueness-beyond-a3}: the global theorem is phrased in terms of pairs of feasible block laws, the fixed-face theorem in terms of the nullspace of the restricted Hessian.
\end{remark}

\begin{theorem}[Differentiability of the maximizer]\label{thm:selector-diff}
Assume $\widetilde J$ is $C^2$ near $(b_0,\xi_0)$, where $\xi_0=s(b_0)$, and that the chart is
 taken on a neighborhood with fixed active support. Assume also that $\partial_\xi \widetilde J(b_0,\xi_0)=0$ and that the fiber Hessian
\begin{equation*}
\partial_{\xi\xi}^2\widetilde J(b_0,\xi_0)
\end{equation*}
is negative definite. Then there exists a neighborhood of $b_0$ in which the local selector is uniquely
defined as the unique local maximizer on the fixed face, it is differentiable, and
\begin{equation*}
Ds(b_0)=
-\bigl[\partial_{\xi\xi}^2\widetilde J(b_0,\xi_0)\bigr]^{-1}
\partial_{\xi b}^2\widetilde J(b_0,\xi_0).
\end{equation*}
\end{theorem}

\begin{proof}
Define
\begin{equation*}
F(b,\xi):=\partial_\xi \widetilde J(b,\xi).
\end{equation*}
Because $\widetilde J$ is $C^2$, the map $F$ is $C^1$ near $(b_0,\xi_0)$. The selector is
characterized in this chart by the first-order condition
\begin{equation*}
F(b,s(b))=0.
\end{equation*}
The Jacobian of $F$ with respect to $\xi$ is
\begin{equation*}
\partial_\xi F(b_0,\xi_0)=\partial_{\xi\xi}^2\widetilde J(b_0,\xi_0),
\end{equation*}
which is nonsingular by hypothesis. Therefore the implicit function theorem yields a unique local
$C^1$ critical-point map $s$ near $b_0$. Because $\partial_{\\xi\\xi}^2\\widetilde J(b_0,\\xi_0)$ is negative definite, continuity of the Hessian implies that, after shrinking the neighborhood if necessary, the critical point $s(b)$ remains a strict local maximizer and is the unique local maximizer on the fixed face near $\\xi_0$. Because the chart is taken on a neighborhood with fixed
active support, no support change occurs inside this parametrization, so the theorem is purely local
on one fixed face of the feasible set.

Differentiating the identity $F(b,s(b))=0$ with respect to $b$ at $b_0$ gives
\begin{equation*}
\partial_{\xi b}^2\widetilde J(b_0,\xi_0)
+
\partial_{\xi\xi}^2\widetilde J(b_0,\xi_0)Ds(b_0)=0,
\end{equation*}
and solving for $Ds(b_0)$ proves the formula.
\end{proof}

Define the optimized value function
\begin{equation}\label{eq:value-function}
V(b):=\widetilde J(b,s(b)).
\end{equation}

\begin{proposition}[Envelope identities]
At the selected point,
\begin{equation*}
DV(b_0)=\partial_b\widetilde J(b_0,\xi_0).
\end{equation*}
Moreover,
\begin{equation*}
D^2V(b_0)
=
\partial_{bb}^2\widetilde J(b_0,\xi_0)
-
\partial_{b\xi}^2\widetilde J(b_0,\xi_0)
\bigl[\partial_{\xi\xi}^2\widetilde J(b_0,\xi_0)\bigr]^{-1}
\partial_{\xi b}^2\widetilde J(b_0,\xi_0).
\end{equation*}
\end{proposition}

\begin{proof}
By the chain rule,
\begin{equation*}
DV(b)=\partial_b\widetilde J(b,s(b))
+
\partial_\xi\widetilde J(b,s(b))Ds(b).
\end{equation*}
The second term vanishes by the selector first-order condition, which proves the first identity.
Differentiating once more and substituting the formula for $Ds(b_0)$ yields the second identity.
\end{proof}

\begin{definition}[Gap functional]
Define
\begin{equation*}
\Delta(b,\xi):=\widetilde J(b,s(b))-\widetilde J(b,\xi)\ge 0.
\end{equation*}
\end{definition}

\begin{theorem}[Quadratic expansion of the gap]\label{thm:gap}
Assume $\widetilde J$ is $C^2$ near $(b_0,\xi_0)$, where $\xi_0=s(b_0)$, and work in the
fixed active-support chart introduced above. Let
\begin{equation*}
K(b_0):=-\partial_{\xi\xi}^2\widetilde J(b_0,\xi_0),
\end{equation*}
and assume that $K(b_0)$ is positive definite. Then
\begin{equation*}
\Delta(b,\xi)
=
\frac12(\xi-s(b))^\top K(b_0)(\xi-s(b))
+
o(\|\xi-s(b)\|^2)
\end{equation*}
as $(b,\xi)\to(b_0,\xi_0)$ within this chart.
\end{theorem}

\begin{proof}
Fix $b$ near $b_0$ and write $\delta:=\xi-s(b)$. Applying Taylor expansion to the map
$\xi\mapsto \widetilde J(b,\xi)$ around $\xi=s(b)$ gives
\begin{equation*}
\widetilde J(b,\xi)
=
\widetilde J(b,s(b))
+
\partial_\xi\widetilde J(b,s(b))\delta
+
\frac12\delta^\top \partial_{\xi\xi}^2\widetilde J(b,s(b))\delta
+
o(\|\delta\|^2).
\end{equation*}
Because $s(b)$ is the local selector in this chart, the first-order term vanishes:
\begin{equation*}
\partial_\xi\widetilde J(b,s(b))=0.
\end{equation*}
Hence
\begin{equation*}
\Delta(b,\xi)
=
-\frac12\delta^\top \partial_{\xi\xi}^2\widetilde J(b,s(b))\delta
+
o(\|\delta\|^2).
\end{equation*}

Now add and subtract $K(b_0)$:
\begin{equation*}
-\partial_{\xi\xi}^2\widetilde J(b,s(b))
=
K(b_0)+\Bigl[-\partial_{\xi\xi}^2\widetilde J(b,s(b))-K(b_0)\Bigr].
\end{equation*}
Since $\widetilde J$ is $C^2$ and $s$ is continuous, the bracketed term tends to zero as
$(b,\xi)\to(b_0,\xi_0)$ within the fixed chart. Therefore
\begin{equation*}
\delta^\top
\Bigl[-\partial_{\xi\xi}^2\widetilde J(b,s(b))-K(b_0)\Bigr]
\delta
=
o(\|\delta\|^2),
\end{equation*}
which yields
\begin{equation*}
\Delta(b,\xi)
=
\frac12\delta^\top K(b_0)\delta
+
o(\|\delta\|^2).
\end{equation*}
This is the claimed expansion.
\end{proof}

\section{Hidden Realizations}

Let $u^\star$ be a selected stationary $(r+1)$-block law with induced active-support kernel $p^\star(a\mid c):=u^\star(c,a)/\eta_{u^\star}(c)$.

\begin{theorem}[Random-mapping realization]
Let $u^\star$ be a selected stationary $(r+1)$-block law, and let $p^\star(a\mid c)$ denote the induced kernel on the active contexts $(c\in A^r\text{ with }\eta_{u^\star}(c)>0)$. Then there exist a measurable map
\begin{equation*}
F:A^r\times [0,1)\to A,
\end{equation*}
an independent and identically distributed (i.i.d.)\ sequence $U_t\sim \text{uniform on }[0,1)$, and an initial context
\begin{equation*}
Y_{-r}^{-1}\sim \eta_{u^\star},
\end{equation*}
independent of $(U_t)_{t\ge 0}$, such that the recursion
\begin{equation*}
Y_t=F(Y_{t-r}^{t-1},U_t),
\qquad t\ge 0,
\end{equation*}
defines a stationary one-sided order-$r$ Markov chain whose stationary $(r+1)$-block law is $u^\star$. In particular, it admits the two-sided stationary extension with the same block law.
\end{theorem}

\begin{proof}
For each context $c\in A^r$, partition $[0,1)$ into half-open intervals
\begin{equation*}
\{I_{c,a}:a\in A\}
\end{equation*}
with lengths
\begin{equation*}
|I_{c,a}|=p^\star(a\mid c).
\end{equation*}
Define
\begin{equation*}
F(c,u)=a
\qquad\text{if }u\in I_{c,a}.
\end{equation*}
Then, for every active context $c$,
\begin{equation*}
P(Y_t=a\mid Y_{t-r}^{t-1}=c)=P(U_t\in I_{c,a})=|I_{c,a}|=p^\star(a\mid c).
\end{equation*}
Thus the recursion has transition kernel $p^\star$.

It remains to identify the stationary law. By definition,
\begin{equation*}
u^\star(c,a)=\eta_{u^\star}(c)p^\star(a\mid c).
\end{equation*}
Because $u^\star$ is stationary-consistent, the context marginal $\eta_{u^\star}$ is invariant for
the induced context chain on $A^r$. Choosing the initial context $Y_{-r}^{-1}$ with law
$\eta_{u^\star}$ therefore makes the resulting order-$r$ Markov chain stationary, and its
stationary $(r+1)$-block law is exactly $u^\star$. For contexts outside the active support,
$F(c,\cdot)$ may be chosen arbitrarily without affecting the realized stationary law.
\end{proof}

\begin{theorem}[Invariance under hidden measure-preserving actions]
Let $(Z_t)$ be a hidden state process, let $(U_t)$ be an independent and identically distributed (i.i.d.)\ sequence of
$\text{uniform on }[0,1)$ variables independent of $(Z_t)$, and let
\begin{equation*}
T_z:[0,1)\to[0,1)
\end{equation*}
be measurable and measure-preserving for each hidden state $z$. Define
\begin{equation*}
\varepsilon_t:=T_{Z_t}(U_t),
\qquad
Y_t:=F(Y_{t-r}^{t-1},\varepsilon_t).
\end{equation*}
Let
\begin{equation*}
\mathcal G_t:=\sigma(Z_s:s\le t),
\qquad
\mathcal F_{t-1}:=\sigma(Y_s:s\le t-1).
\end{equation*}
Then
\begin{equation*}
\mathcal L(\varepsilon_t\mid \mathcal G_t\vee \mathcal F_{t-1})=\text{uniform on }[0,1).
\end{equation*}
Consequently,
\begin{equation*}
\mathcal L(\varepsilon_t\mid \mathcal G_t)=\text{uniform on }[0,1),
\qquad
\mathcal L(\varepsilon_t\mid \mathcal F_{t-1})=\text{uniform on }[0,1),
\end{equation*}
and the induced visible kernel remains $p^\star$.
\end{theorem}

\begin{proof}
For every Borel set $B\subseteq[0,1)$, define
\begin{equation*}
A_B:=T_{Z_t}^{-1}(B).
\end{equation*}
Then $A_B$ is $\mathcal G_t$-measurable, hence $\mathcal G_t\vee\mathcal F_{t-1}$-measurable. Since $U_t$ is independent of $\mathcal G_t\vee\mathcal F_{t-1}$ and is uniform on $[0,1)$, the conditional-expectation identity for indicators of random measurable sets yields
\begin{equation*}
P\bigl(U_t\in A_B\mid \mathcal G_t\vee \mathcal F_{t-1}\bigr)=\lambda(A_B)
\qquad\text{a.s.}
\end{equation*}
Because $\varepsilon_t=T_{Z_t}(U_t)$, this becomes
\begin{equation*}
P(\varepsilon_t\in B\mid \mathcal G_t\vee \mathcal F_{t-1})
=\lambda\bigl(T_{Z_t}^{-1}(B)\bigr).
\end{equation*}
Since each $T_z$ preserves Lebesgue measure,
\begin{equation*}
\lambda\bigl(T_{Z_t}^{-1}(B)\bigr)=\lambda(B).
\end{equation*}
Therefore
\begin{equation*}
\mathcal L(\varepsilon_t\mid \mathcal G_t\vee \mathcal F_{t-1})=\text{uniform on }[0,1).
\end{equation*}
The two displayed marginal conditional-law statements follow immediately by the tower property.

Finally, for every context $c\in A^r$ and symbol $a\in A$,
\begin{equation*}
P(Y_t=a\mid Y_{t-r}^{t-1}=c)
=
P(\varepsilon_t\in I_{c,a}\mid Y_{t-r}^{t-1}=c)
=
|I_{c,a}|=p^\star(a\mid c).
\end{equation*}
Thus hidden structural action is compatible with no change in the visible law while leaving the visible
transition kernel unchanged.
\end{proof}

\begin{remark}
The theorem shows that hidden structure may act on the primitive shocks in a measure-preserving way while remaining invisible at the level of the conditional visible law. In particular, this section proves a realizability and invariance statement for the selected visible kernel, not a realization theorem inside a prescribed hidden observational fiber.
\end{remark}

\section{Empirical Maximizers and Consistency}

Let
\begin{equation*}
Y_0,\dots,Y_n
\end{equation*}
be one observed visible sample path on the finite alphabet $A$. For $c\in A^r$ and $a\in A$,
define the empirical $(r+1)$-block frequency
\begin{equation}\label{eq:empirical-block}
\hat u_n(c,a):=\frac{1}{n-r}\sum_{t=r}^{n-1}\mathbf 1\{Y_{t-r}^{t-1}=c,\ Y_t=a\}.
\end{equation}
Its empirical context marginal is
\begin{equation*}
\hat\eta_n(c):=\sum_{a\in A}\hat u_n(c,a).
\end{equation*}
If the experiment-induced features are $G_1,\dots,G_m$, define the empirical targets by
\begin{equation*}
\hat b_{n,j}:=\sum_{c,a}\hat u_n(c,a)G_j(c,a),
\qquad j=1,\dots,m.
\end{equation*}

For a deterministic target vector $b\in\mathbb R^m$, write
\begin{equation*}
\mathcal U(b):=
\Bigl\{u\text{ on }A^{r+1}:u\text{ is stationary-consistent and }
\sum_{c,a}u(c,a)G_j(c,a)=b_j,\ j=1,\dots,m
\Bigr\}.
\end{equation*}
Then the empirical feasible class is
\begin{equation*}
\hat{\mathcal U}_n=\mathcal U(\hat b_n).
\end{equation*}
Define the empirical criterion
\begin{equation*}
\hat J_n(u):=-\sum_{c,a}u(c,a)\log\frac{u(c,a)}{\eta_u(c)}.
\end{equation*}
Because the objective depends only on $u$ and not explicitly on $n$, one may equivalently write
\begin{equation*}
\hat J_n(u)=J(u)
\qquad\text{on }\hat{\mathcal U}_n.
\end{equation*}
Let
\begin{equation}\label{eq:empirical-selector}
\hat u_n^\star\in\arg\max_{u\in\hat{\mathcal U}_n}J(u)
\end{equation}
be any empirical selector.

\begin{proposition}[Empirical block-frequency convergence]\label{prop:block-freq}
Assume the true visible process is stationary, irreducible, and $r$-step Markov on the finite
alphabet $A$. Then for each $(c,a)\in A^{r+1}$,
\begin{equation*}
\hat u_n(c,a)\to u_0(c,a)
\qquad\text{almost surely,}
\end{equation*}
where $u_0$ denotes the true stationary $(r+1)$-block law. Consequently,
\begin{equation*}
\hat\eta_n(c)\to \eta_{u_0}(c)
\qquad\text{almost surely}
\end{equation*}
for every context $c\in A^r$.
\end{proposition}

\begin{proof}
For each fixed block $(c,a)$, the indicator
\begin{equation*}
\mathbf 1\{Y_{t-r}^{t-1}=c,\ Y_t=a\}
\end{equation*}
is bounded. The ergodic theorem for finite-state Markov chains therefore gives
\begin{equation*}
\hat u_n(c,a)\to u_0(c,a)
\qquad\text{almost surely}
\end{equation*}
for each $(c,a)$. Summing over $a\in A$ yields the convergence of the empirical context
marginals.
\end{proof}

\begin{proposition}[Empirical feature convergence]\label{prop:feature-convergence}
Assume the hypotheses of \cref{prop:block-freq}. Then for each $j=1,\dots,m$,
\begin{equation*}
\hat b_{n,j}\to b_{0,j}:=\sum_{c,a}u_0(c,a)G_j(c,a)
\qquad\text{almost surely.}
\end{equation*}
Equivalently,
\begin{equation*}
\hat b_n\to b_0
\qquad\text{almost surely.}
\end{equation*}
\end{proposition}

\begin{proof}
Because $A^{r+1}$ is finite, each feature $G_j$ is bounded on $A^{r+1}$. Hence
\begin{equation*}
\hat b_{n,j}-b_{0,j}
=
\sum_{c,a}\bigl(\hat u_n(c,a)-u_0(c,a)\bigr)G_j(c,a),
\end{equation*}
and every summand converges almost surely to zero by \cref{prop:block-freq}. Therefore
$\hat b_{n,j}\to b_{0,j}$ almost surely for each $j$, and thus $\hat b_n\to b_0$ almost surely.
\end{proof}

Fix the fixed-support face
\begin{equation*}
\mathfrak F\subseteq \Delta(A^{r+1})
\end{equation*}
containing the selected point $u^\star$ in its relative interior, and let
\begin{equation*}
\mathcal A_{\mathfrak F}
\end{equation*}
denote the affine subspace of stationary-consistent block laws supported on $\mathfrak F$. Write
\begin{equation*}
\mathrm{Tan}(\mathcal A_{\mathfrak F})
\end{equation*}
for the translation space of the affine set $\mathcal A_{\mathfrak F}$. Define the linear moment map
\begin{equation*}
T:\mathcal A_{\mathfrak F}\to\mathbb R^m,
\qquad
T(u):=\Bigl(\sum_{c,a}u(c,a)G_j(c,a)\Bigr)_{j=1}^m.
\end{equation*}

\begin{lemma}[Local feasible continuation near the selected point]\label{lem:local-feasible}
Assume $u^\star\in \mathcal U(b_0)\cap \mathcal A_{\mathfrak F}$ lies in the relative interior of the
face $\mathfrak F$, and assume that the restriction of $T$ to the translation space of
$\mathcal A_{\mathfrak F}$ has full row rank. Then there exist a neighborhood $B$ of $b_0$ in
$\mathbb R^m$, a neighborhood $N$ of $u^\star$ in $\mathcal A_{\mathfrak F}$, and a constant $L>0$ such that
\begin{equation*}
\overline N\subset \operatorname{relint}(\mathfrak F),
\end{equation*}
where $\operatorname{relint}$ denotes relative interior, and for every $b\in B$ there exists a point
\begin{equation*}u(b)\in \mathcal U(b)\cap N
\end{equation*}
satisfying
\begin{equation*}
\|u(b)-u^\star\|\le L\|b-b_0\|.
\end{equation*}
In particular,
\begin{equation*}
\mathcal U(b)\cap N\neq\varnothing
\end{equation*}
for every $b\in B$.
\end{lemma}

\begin{proof}
Because the restricted moment map has full row rank, there exists a linear right inverse
\begin{equation*}
R:\mathbb R^m\to \mathrm{Tan}(\mathcal A_{\mathfrak F})
\end{equation*}
with $TR=I_m$. Define
\begin{equation*}u(b):=u^\star+R(b-b_0).
\end{equation*}
Then $u(b)\in \mathcal A_{\mathfrak F}$ and
\begin{equation*}
T\bigl(u(b)\bigr)=T(u^\star)+TR(b-b_0)=b.
\end{equation*}
Thus $u(b)\in \mathcal U(b)\cap \mathcal A_{\mathfrak F}$ whenever it remains in the face.

Since $u^\star\in \operatorname{relint}(\mathfrak F)$, there exists a neighborhood $N$ of $u^\star$
in $\mathcal A_{\mathfrak F}$ whose compact closure satisfies
\begin{equation*}
\overline N\subset \operatorname{relint}(\mathfrak F).
\end{equation*}
Because $u(b)\to u^\star$ as $b\to b_0$, there exists a neighborhood $B$ of $b_0$ such that
$u(b)\in N$ for every $b\in B$. This gives the claim, with $L:=\|R\|$.
\end{proof}

\begin{theorem}[Local consistency on a fixed-support face]
Assume the hypotheses of \cref{prop:block-freq,prop:feature-convergence}. Assume also that:
\begin{enumerate}[label=(C\arabic*),leftmargin=2.5em]
\item the population selector $u^\star$ is the unique maximizer of $J$ over
\begin{equation*}
\mathcal U(b_0)\cap \mathcal A_{\mathfrak F},
\end{equation*}

\item $u^\star$ lies in the relative interior of the fixed-support face $\mathfrak F$,

\item the restricted moment map in \cref{lem:local-feasible} has full row rank.
\end{enumerate}
Then there exists a neighborhood $N$ of $u^\star$ in $\mathcal A_{\mathfrak F}$ with compact closure contained in $\operatorname{relint}(\mathfrak F)$ such that the following holds almost surely for all sufficiently large $n$: the compact local empirical feasible class $\hat{\mathcal U}_n\cap \overline N$ is nonempty, and every empirical maximizer
\begin{equation*}
\hat u_n^\star\in \arg\max_{u\in \hat{\mathcal U}_n\cap \overline N}J(u)
\end{equation*}
satisfies $\hat u_n^\star\to u^\star$ almost surely.
\end{theorem}

\begin{proof}
Because $u^\star$ is the unique maximizer of the continuous function $J$ on the compact set
\begin{equation*}
\mathcal U(b_0)\cap \mathcal A_{\mathfrak F},
\end{equation*}
there exists a neighborhood $N$ of $u^\star$ in $\mathcal A_{\mathfrak F}$ with compact closure
\begin{equation*}
\overline N\subset \operatorname{relint}(\mathfrak F)
\end{equation*}
and an $\varepsilon>0$ such that
\begin{equation*}
J(u)\le J(u^\star)-\varepsilon
\qquad
\text{for every }u\in \bigl(\mathcal U(b_0)\cap \mathcal A_{\mathfrak F}\bigr)\setminus N.
\end{equation*}
Shrinking $N$ if necessary, I may also assume that \cref{lem:local-feasible} holds on this same
neighborhood.

By \cref{prop:feature-convergence}, one has~\cite{billingsley1995,vandervaart1998}
\begin{equation*}
\hat b_n\to b_0
\qquad\text{almost surely.}
\end{equation*}
Hence, by \cref{lem:local-feasible}, for all sufficiently large $n$ there exists
\begin{equation*}u_n\in \hat{\mathcal U}_n\cap N
\end{equation*}
with $u_n\to u^\star$. In particular,
\begin{equation*}
\hat{\mathcal U}_n\cap N\neq\varnothing
\end{equation*}
eventually almost surely.

Fix an almost sure sample point on which $\hat b_n\to b_0$. Let $\hat u_n^\star$ be any local
empirical maximizer in $\hat{\mathcal U}_n\cap N$. Because $\overline N$ is compact, every
subsequence of $(\hat u_n^\star)$ has a further convergent subsequence, let
\begin{equation*}
\hat u_{n_k}^\star\to \bar u\in \overline N.
\end{equation*}
Since
\begin{equation*}
T(\hat u_{n_k}^\star)=\hat b_{n_k}
\qquad\text{and}\qquad
\hat b_{n_k}\to b_0,
\end{equation*}
continuity of $T$ yields
\begin{equation*}
T(\bar u)=b_0,
\end{equation*}
so $\bar u\in \mathcal U(b_0)\cap \overline N$.

Moreover, because $u_{n_k}\in \hat{\mathcal U}_{n_k}\cap N$ and $\hat u_{n_k}^\star$ maximizes
$J$ over that local feasible class,
\begin{equation*}
J(\hat u_{n_k}^\star)\ge J(u_{n_k}).
\end{equation*}
Passing to the limit and using continuity of $J$ gives
\begin{equation*}
J(\bar u)\ge J(u^\star).
\end{equation*}
Since $u^\star$ is feasible for the population problem and is the unique maximizer of $J$ on
$\mathcal U(b_0)\cap \mathcal A_{\mathfrak F}$, it follows that $\bar u=u^\star$.

Thus every convergent subsequence of $(\hat u_n^\star)$ converges to $u^\star$, and therefore the
whole sequence converges almost surely to $u^\star$.
\end{proof}

The local asymptotic statistic from \cref{thm:gap} may be evaluated at empirical local
coordinates derived from $\hat u_n^\star$. Thus the gap functional is a population-geometric object and may serve as the basis of a testing procedure under additional asymptotic assumptions.

\begin{remark}
A full empirical theory across support changes would require substantially more work. The present
result isolates the local argument: convergence of block frequencies yields convergence of
retained moments, full row rank yields nearby feasible continuation on one face, and uniqueness of
the population selector yields the separation needed to force local empirical maximizers back to
$u^\star$.
\end{remark}

\section{An Aliased Hidden-State Example}

This section gives a concrete application in which the selector resolves a genuinely underidentified visible model, while hidden completion remains non-unique. The construction stays within the finite-state finite-memory scope of the paper and makes the observational-fiber interpretation explicit.

Let the hidden state space be
\begin{equation*}
E:=\{a_0,a_1,b_0,b_1\},
\end{equation*}
and let the visible alphabet be $A:=\{0,1\}$. Define the observation map
\begin{equation*}
\phi(a_0)=\phi(a_1)=0,
\qquad
\phi(b_0)=\phi(b_1)=1.
\end{equation*}
Thus two hidden states are observationally aliased into each visible symbol.

Fix parameters
\begin{equation*}
a,b,\lambda,\mu\in(0,1)
\end{equation*}
and define a Markov chain $(X_t)$ on $E$ by the transition matrix
\begin{equation*}
K=
\begin{pmatrix}
(1-a)\lambda & (1-a)(1-\lambda) & a\mu & a(1-\mu)\\
(1-a)\lambda & (1-a)(1-\lambda) & a\mu & a(1-\mu)\\
b\lambda & b(1-\lambda) & (1-b)\mu & (1-b)(1-\mu)\\
b\lambda & b(1-\lambda) & (1-b)\mu & (1-b)(1-\mu)
\end{pmatrix},
\end{equation*}
where the rows are ordered as $(a_0,a_1,b_0,b_1)$. Let the visible process be
\begin{equation*}
Y_t:=\phi(X_t).
\end{equation*}

The point of the construction is that the visible transition probabilities depend only on the aggregated parameters $a$ and $b$, whereas the hidden completion still depends on the latent splitting parameters $\lambda$ and $\mu$.

\begin{proposition}[Aliased hidden realization of any binary first-order chain]\label{prop:aliased-binary-realization}
Let
\begin{equation*}
P=
\begin{pmatrix}
1-a & a\\
b & 1-b
\end{pmatrix}
\qquad\text{with }a,b\in(0,1)
\end{equation*}
be a stationary irreducible binary Markov kernel, and let
\begin{equation*}
m:=\frac{a}{a+b}.
\end{equation*}
For any $\lambda,\mu\in(0,1)$, the hidden transition matrix $K$ defined above has the following properties:
\begin{enumerate}[label=(\roman*),leftmargin=2.5em]
\item The hidden chain is irreducible and aperiodic.
\item It admits the stationary distribution
\begin{equation*}
\pi^E=\bigl((1-m)\lambda,(1-m)(1-\lambda),m\mu,m(1-\mu)\bigr).
\end{equation*}
\item The observed process $Y_t=\phi(X_t)$ is a stationary binary first-order Markov chain with transition matrix $P$.
\end{enumerate}
In particular, every stationary irreducible binary first-order Markov chain admits infinitely many aliased hidden-state realizations of this form.
\end{proposition}

\begin{proof}
Because all entries of $K$ are strictly positive, the hidden chain is irreducible and aperiodic. This proves \textnormal{(i)}.

To prove \textnormal{(ii)}, write
\begin{equation*}
\pi^E=(\pi_{a_0},\pi_{a_1},\pi_{b_0},\pi_{b_1})
:=\bigl((1-m)\lambda,(1-m)(1-\lambda),m\mu,m(1-\mu)\bigr).
\end{equation*}
It is a probability vector because
\begin{equation*}
(1-m)\lambda+(1-m)(1-\lambda)+m\mu+m(1-\mu)=1.
\end{equation*}
Now compute the first coordinate of $\pi^E K$:
\begin{align*}
(\pi^E K)_{a_0}
&=\bigl(\pi_{a_0}+\pi_{a_1}\bigr)(1-a)\lambda+\bigl(\pi_{b_0}+\pi_{b_1}\bigr)b\lambda\\
&=(1-m)(1-a)\lambda+mb\lambda.
\end{align*}
Because $m=a/(a+b)$, one has
\begin{equation*}
(1-m)a=mb,
\end{equation*}
so
\begin{equation*}
(1-m)(1-a)+mb=(1-m)- (1-m)a+mb=1-m.
\end{equation*}
Hence
\begin{equation*}
(\pi^E K)_{a_0}=(1-m)\lambda=\pi_{a_0}.
\end{equation*}
The same calculation gives
\begin{equation*}
(\pi^E K)_{a_1}=(1-m)(1-\lambda)=\pi_{a_1}.
\end{equation*}
Likewise,
\begin{align*}
(\pi^E K)_{b_0}
&=\bigl(\pi_{a_0}+\pi_{a_1}\bigr)a\mu+\bigl(\pi_{b_0}+\pi_{b_1}\bigr)(1-b)\mu\\
&=(1-m)a\mu+m(1-b)\mu
= m\mu
=\pi_{b_0},
\end{align*}
where the identity $(1-m)a=mb$ was used again. The fourth coordinate is analogous, so
\begin{equation*}
\pi^E K=\pi^E.
\end{equation*}
Thus $\pi^E$ is stationary.

To prove \textnormal{(iii)}, note that the visible transition probabilities are obtained by summing the hidden transition probabilities over each fiber of $\phi$. If the current hidden state is $a_0$, then the probability of moving to the visible symbol $0$ is
\begin{equation*}
K(a_0,a_0)+K(a_0,a_1)=(1-a)\lambda+(1-a)(1-\lambda)=1-a,
\end{equation*}
and the probability of moving to the visible symbol $1$ is
\begin{equation*}
K(a_0,b_0)+K(a_0,b_1)=a\mu+a(1-\mu)=a.
\end{equation*}
Exactly the same identities hold when the current hidden state is $a_1$, because the first two rows of $K$ are equal. Similarly, if the current hidden state is $b_0$ or $b_1$, then the last two rows of $K$ give
\begin{equation*}
P(Y_{t+1}=0\mid X_t=x)=b,
\qquad
P(Y_{t+1}=1\mid X_t=x)=1-b,
\qquad x\in\{b_0,b_1\}.
\end{equation*}
Consequently, for each visible state $y\in\{0,1\}$ and each symbol $z\in\{0,1\}$, the quantity
\begin{equation*}
P(Y_{t+1}=z\mid X_t=x)
\end{equation*}
is the same for all hidden states $x$ satisfying $\phi(x)=y$, and its common value is exactly $P(y,z)$. By conditioning on the sigma-field generated by $Y_t$ and applying the tower property, it follows that
\begin{equation*}
P(Y_{t+1}=z\mid Y_t=y)=P(y,z).
\end{equation*}
Hence $(Y_t)$ is a binary first-order Markov chain with transition matrix $P$. Since the hidden chain is stationary under $\pi^E$, the visible process is stationary as well. This proves \textnormal{(iii)}.

Finally, varying $\lambda$ and $\mu$ in $(0,1)$ changes the hidden transition matrix $K$ while leaving the visible transition matrix $P$ unchanged. Hence one obtains infinitely many distinct aliased hidden-state realizations of the same visible chain.
\end{proof}

\begin{corollary}[Visible underidentification under a fixed mean]\label{cor:aliased-fixed-mean}
Fix $m\in(0,1)$. For each
\begin{equation*}
q\in\bigl(0,\min\{m,1-m\}\bigr),
\end{equation*}
define
\begin{equation*}
a(q):=\frac{q}{1-m},
\qquad
b(q):=\frac{q}{m}.
\end{equation*}
Then for every choice of $\lambda,\mu\in(0,1)$, the aliased hidden transition matrix obtained from
\begin{equation*}
a=a(q),
\qquad
b=b(q)
\end{equation*}
induces a stationary binary visible first-order Markov chain with stationary mean $m$. As $q$ varies, these visible laws are distinct. Consequently, if the retained visible observable is only the stationary mean $m$, then the visible completion problem is non-unique even inside this single aliased hidden-state architecture.
\end{corollary}

\begin{proof}
For each admissible $q$, the definitions of $a(q)$ and $b(q)$ give
\begin{equation*}
(1-m)a(q)=q=mb(q),
\end{equation*}
so the stationary mean of the visible chain in \cref{prop:aliased-binary-realization} is exactly $m$. Distinct values of $q$ give distinct transition matrices
\begin{equation*}
\begin{pmatrix}
1-a(q) & a(q)\\
b(q) & 1-b(q)
\end{pmatrix},
\end{equation*}
hence distinct visible laws. Therefore the visible law is not determined by the retained mean alone.
\end{proof}

\begin{theorem}[Visible maximizing completion under aliasing]\label{thm:aliased-canonical-completion}
Fix $m\in(0,1)$, and let $\mathcal V_m^{\mathrm{aliased}}$ denote the class of stationary irreducible binary visible first-order Markov laws generated by the aliased hidden-state construction of \cref{prop:aliased-binary-realization,cor:aliased-fixed-mean} under the single retained visible observable
\begin{equation*}
P(Y_t=1)=m.
\end{equation*}
Then the entropy-rate functional has a unique maximizer on $\mathcal V_m^{\mathrm{aliased}}$. This unique selector is the i.i.d.\ Bernoulli$(m)$ law, equivalently the binary first-order Markov chain with transition matrix
\begin{equation*}
P^\star=
\begin{pmatrix}
1-m & m\\
1-m & m
\end{pmatrix}.
\end{equation*}
Thus the selector canonically resolves the visible underidentification inside this aliased hidden-state architecture, even though the hidden completion remains non-unique in the explicit sense of \cref{prop:no-canonical-hidden-completion}, equivalently, the selected visible law still has an infinite hidden observational fiber by \cref{cor:selected-visible-fiber-continuum}.
\end{theorem}

\begin{proof}
By \cref{cor:aliased-fixed-mean}, the class $\mathcal V_m^{\mathrm{aliased}}$ is parameterized by
\begin{equation*}
q\in\bigl(0,\min\{m,1-m\}\bigr)
\end{equation*}
through
\begin{equation*}
a(q)=\frac{q}{1-m},
\qquad
b(q)=\frac{q}{m}.
\end{equation*}
For such a chain the entropy rate is
\begin{equation*}
h(q)=(1-m)h_2\!\Bigl(\frac{q}{1-m}\Bigr)+mh_2\!\Bigl(\frac{q}{m}\Bigr),
\end{equation*}
where
\begin{equation*}
h_2(p):=-p\log p-(1-p)\log(1-p).
\end{equation*}
Since
\begin{equation*}
h_2'(p)=\log\frac{1-p}{p}
\qquad\text{and}\qquad
h_2''(p)=-\frac{1}{p(1-p)}<0
\qquad (0<p<1),
\end{equation*}
one obtains
\begin{align*}
h'(q)
&=h_2'\!\Bigl(\frac{q}{1-m}\Bigr)+h_2'\!\Bigl(\frac{q}{m}\Bigr)\\
&=\log\frac{1-m-q}{q}+\log\frac{m-q}{q}\\
&=\log\frac{(1-m-q)(m-q)}{q^2},
\end{align*}
and therefore
\begin{equation*}
h''(q)=-\frac{1}{1-m-q}-\frac{1}{m-q}-\frac{2}{q}<0
\qquad\text{for }q\in\bigl(0,\min\{m,1-m\}\bigr).
\end{equation*}
Hence $h$ is strictly concave on the whole parameter interval, so it has at most one maximizer.

The critical-point equation $h'(q)=0$ is equivalent to
\begin{equation*}
(1-m-q)(m-q)=q^2.
\end{equation*}
Expanding the left-hand side gives
\begin{equation*}
m(1-m)-q+q^2=q^2,
\end{equation*}
so the unique critical point is
\begin{equation*}
q^\star=m(1-m).
\end{equation*}
Because
\begin{equation*}
0<m(1-m)<\min\{m,1-m\}
\qquad\text{for every }m\in(0,1),
\end{equation*}
this critical point lies in the admissible interval. Strict concavity now implies that $q^\star$ is the unique global maximizer of $h$ on $\mathcal V_m^{\mathrm{aliased}}$.

Substituting $q^\star$ into the parameterization yields
\begin{equation*}
a^\star=\frac{q^\star}{1-m}=m,
\qquad
b^\star=\frac{q^\star}{m}=1-m.
\end{equation*}
Therefore the unique maximizing visible kernel is
\begin{equation*}
P^\star=
\begin{pmatrix}
1-m & m\\
1-m & m
\end{pmatrix},
\end{equation*}
which is the i.i.d.\ Bernoulli$(m)$ law. This proves the visible-selection statement. The hidden non-uniqueness assertion is then made precise by \cref{prop:no-canonical-hidden-completion,cor:selected-visible-fiber-continuum}.
\end{proof}

\begin{proposition}[Failure of hidden maximizing completion]\label{prop:no-canonical-hidden-completion}
Fix $m\in(0,1)$ and let $P^\star$ be the selected visible kernel from \cref{thm:aliased-canonical-completion}. Then the class of aliased hidden-state realizations that generate $P^\star$ contains infinitely many distinct hidden transition matrices, namely
\begin{equation*}
K_{\lambda,\mu}=
\begin{pmatrix}
(1-m)\lambda & (1-m)(1-\lambda) & m\mu & m(1-\mu)\\
(1-m)\lambda & (1-m)(1-\lambda) & m\mu & m(1-\mu)\\
(1-m)\lambda & (1-m)(1-\lambda) & m\mu & m(1-\mu)\\
(1-m)\lambda & (1-m)(1-\lambda) & m\mu & m(1-\mu)
\end{pmatrix},
\qquad \lambda,\mu\in(0,1).
\end{equation*}
Consequently, the selector of the present paper determines a canonical visible completion, but it does not determine a canonical hidden completion.
\end{proposition}

\begin{proof}
Set $a=m$ and $b=1-m$ in \cref{prop:aliased-binary-realization}. Then the resulting visible kernel is exactly
\begin{equation*}
P^\star=
\begin{pmatrix}
1-m & m\\
1-m & m
\end{pmatrix}.
\end{equation*}
For every pair $(\lambda,\mu)\in(0,1)^2$, \cref{prop:aliased-binary-realization} shows that the hidden transition matrix $K_{\lambda,\mu}$ realizes this same visible law.

It remains to show that different parameter pairs produce different hidden completions. If
\begin{equation*}
(\lambda,\mu)\ne (\lambda',\mu'),
\end{equation*}
then either $\lambda\ne\lambda'$ or $\mu\ne\mu'$. In the first case the entry in position $(a_0,a_0)$ differs:
\begin{equation*}
K_{\lambda,\mu}(a_0,a_0)=(1-m)\lambda
\ne (1-m)\lambda'=K_{\lambda',\mu'}(a_0,a_0).
\end{equation*}
In the second case the entry in position $(a_0,b_0)$ differs:
\begin{equation*}
K_{\lambda,\mu}(a_0,b_0)=m\mu
\ne m\mu'=K_{\lambda',\mu'}(a_0,b_0).
\end{equation*}
Hence
\begin{equation*}
K_{\lambda,\mu}\ne K_{\lambda',\mu'}.
\end{equation*}
Since $(0,1)^2$ contains infinitely many points, this yields infinitely many distinct hidden transition matrices realizing the same selected visible law.

Therefore the entropy selector canonically resolves the visible ambiguity but leaves the hidden completion problem non-unique. This is exactly the distinction between visible selection and hidden completion that motivates the observational-fiber viewpoint of the paper.
\end{proof}

\begin{proposition}[Uniqueness of the visible maximizer and nonuniqueness of hidden realizations]
Fix $m\in(0,1)$. Within the aliased hidden-state architecture of this section, the retained visible observable
\begin{equation*}
P(Y_t=1)=m
\end{equation*}
determines a unique canonical visible completion by \cref{thm:aliased-canonical-completion}, but it does not determine a unique hidden completion by \cref{prop:no-canonical-hidden-completion}. Consequently, any rule that selects a single hidden law from this architecture must use additional information not contained in the retained visible observable and the observation map alone.
\end{proposition}

\begin{proof}
The first assertion is exactly the content of \cref{thm:aliased-canonical-completion}. The second assertion is exactly the content of \cref{prop:no-canonical-hidden-completion}. Therefore the observation experiment together with the retained visible mean determines a unique selected visible law but leaves infinitely many admissible hidden realizations.

Suppose, to the contrary, that there were a hidden-selection rule on this architecture determined solely by the observation map and the retained visible observable. Applied at the present value of $m$, such a rule would have to choose one hidden law from a family of infinitely many hidden laws that are indistinguishable at the retained visible level. Hence the rule would require a tie-breaking criterion not encoded in the given visible information. This contradicts the assumption that the rule is determined by the observation map and retained visible observable alone.
\end{proof}

\begin{corollary}[Continuum hidden observational fiber at the selected visible law]\label{cor:selected-visible-fiber-continuum}
Let $\nu^\star$ denote the stationary visible path law of the selected Bernoulli$(m)$ process from \cref{thm:aliased-canonical-completion}, and let $\Pi$ be the pathwise observation map induced by $\phi$. For each $(\lambda,\mu)\in(0,1)^2$, let $Q_{\lambda,\mu}$ be the stationary path law of the hidden Markov chain with transition matrix $K_{\lambda,\mu}$ from \cref{prop:no-canonical-hidden-completion}. Then
\begin{equation*}
Q_{\lambda,\mu}\in \mathcal E_\Pi(\nu^\star)
\qquad\text{for every }(\lambda,\mu)\in(0,1)^2,
\end{equation*}
and the set
\begin{equation*}
\{Q_{\lambda,\mu}:(\lambda,\mu)\in(0,1)^2\}
\end{equation*}
is infinite. In particular, the observational fiber of the selected visible law is non-singleton and in fact contains infinitely many stationary hidden laws.
\end{corollary}

\begin{proof}
By \cref{prop:no-canonical-hidden-completion}, every matrix $K_{\lambda,\mu}$ generates the same selected visible kernel $P^\star$. Because each such hidden chain is finite-state, irreducible, and aperiodic, it has a unique stationary path law, denoted here by $Q_{\lambda,\mu}$. The visible process obtained from $Q_{\lambda,\mu}$ under the pathwise observation map $\Pi$ is exactly the stationary Bernoulli$(m)$ law $\nu^\star$. Hence
\begin{equation*}
\Pi_\#Q_{\lambda,\mu}=\nu^\star,
\end{equation*}
so by the definition of the observational fiber in \cref{sec:observational-fibers}, one has
\begin{equation*}
Q_{\lambda,\mu}\in \mathcal E_\Pi(\nu^\star).
\end{equation*}
If $(\lambda,\mu)\neq(\lambda',\mu')$, then \cref{prop:no-canonical-hidden-completion} gives
\begin{equation*}
K_{\lambda,\mu}\neq K_{\lambda',\mu'}.
\end{equation*}
Because every state has strictly positive stationary mass, the stationary path law of a finite-state Markov chain determines its two-point cylinder probabilities and hence its transition matrix. Therefore distinct transition matrices yield distinct stationary path laws, so
\begin{equation*}
Q_{\lambda,\mu}\neq Q_{\lambda',\mu'}.
\end{equation*}
Therefore the family above is infinite.
\end{proof}

\subsection{Comparison with Related Approaches}

The selector is related to three nearby objectives. It concerns canonical visible completion rather than hidden identification. It is not obtained by projection alone from a hidden chain, but by entropy-rate maximization on a visible feasible class determined by retained observables. The exponential representation appears, when it appears, as a consequence of the constrained block-law optimization itself.

\begin{proposition}[Comparison with hidden identification]
Fix $m\in(0,1)$ and work in the aliased hidden-state architecture of \cref{prop:aliased-binary-realization}. Then the following statements hold.
\begin{enumerate}[label=(\roman*),leftmargin=2.5em]
\item The retained visible observable $P(Y_t=1)=m$ does not identify a unique hidden law.
\item The selected visible law is not determined by projection alone from the retained visible observable: within the same aliased architecture and under the same retained mean $m$, there exist infinitely many distinct visible first-order Markov laws, but exactly one entropy-rate maximizer.
\end{enumerate}
Consequently, the selector of the present paper solves a canonical visible-completion problem and not a hidden-identification problem.
\end{proposition}

\begin{proof}
Assertion \textnormal{(i)} is exactly the content of \cref{cor:selected-visible-fiber-continuum}: the selected visible law $\nu^\star$ has an infinite observational fiber
\begin{equation*}
\mathcal E_\Pi(\nu^\star).
\end{equation*}
In particular, the retained visible observable together with the observation map does not determine a unique hidden stationary law.

For \textnormal{(ii)}, \cref{cor:aliased-fixed-mean} shows that under the fixed retained mean $m$ there are infinitely many distinct visible first-order Markov laws in the same aliased architecture, parameterized by
\begin{equation*}
q\in\bigl(0,\min\{m,1-m\}\bigr).
\end{equation*}
By \cref{thm:aliased-canonical-completion}, the entropy-rate functional has a unique maximizer on this visible class, attained at
\begin{equation*}
q^\star=m(1-m),
\end{equation*}
which yields the Bernoulli$(m)$ law. Therefore the selected visible law is not obtained from the retained mean by projection alone, because the same retained mean is compatible with infinitely many projected visible laws. What singles out one visible law is the additional canonical entropy-rate selection principle.

Combining \textnormal{(i)} and \textnormal{(ii)} proves the final statement.
\end{proof}

\begin{proposition}[Comparison with exponential-family selection]
Assume the hypotheses of \cref{prop:kkt}. Then any exponential representation of the selected active-support kernel in the present paper is derived from the constrained visible optimization problem itself. More precisely, on the active support one has the multiplier identity
\begin{equation*}
\log\frac{u^\star(c,a)}{\eta_{u^\star}(c)}
=
-\gamma
-\sum_{j=1}^m \lambda_j G_j(c,a)
-\psi(c)
+\psi(\sigma(c,a)),
\end{equation*}
so the selected kernel is determined by the block-law KKT system. In particular, the paper does not begin by specifying an external exponential-family transition model on the hidden side.
\end{proposition}

\begin{proof}
The displayed identity is exactly the conclusion of \cref{prop:kkt}. Therefore any exponential expression for the selected visible kernel is obtained from the Lagrange multipliers attached to the normalization, stationarity, and retained-observable constraints of the visible block-law problem.

It follows that the exponential structure is derived rather than postulated. In particular, in the aliased hidden-state application of this paper, \cref{prop:no-canonical-hidden-completion,cor:selected-visible-fiber-continuum} show that even the selected visible law is compatible with infinitely many distinct hidden laws. Hence an externally imposed hidden exponential-family ansatz would add modeling structure not determined by the retained visible observable and observation map alone.
\end{proof}

\begin{proposition}[Comparison with hidden-entropy maximization]\label{prop:contrast-hidden-entropy}
Fix $m\in(0,1)$ and consider the family of hidden completions
\begin{equation*}
K_{\lambda,\mu}=
\begin{pmatrix}
(1-m)\lambda & (1-m)(1-\lambda) & m\mu & m(1-\mu)\\
(1-m)\lambda & (1-m)(1-\lambda) & m\mu & m(1-\mu)\\
(1-m)\lambda & (1-m)(1-\lambda) & m\mu & m(1-\mu)\\
(1-m)\lambda & (1-m)(1-\lambda) & m\mu & m(1-\mu)
\end{pmatrix},
\qquad (\lambda,\mu)\in(0,1)^2,
\end{equation*}
which all generate the selected visible Bernoulli$(m)$ law from \cref{thm:aliased-canonical-completion}. Then the hidden entropy rate of the stationary hidden chain equals
\begin{equation*}
H_{\mathrm{hid}}(\lambda,\mu)
=h_2(m)+(1-m)h_2(\lambda)+mh_2(\mu).
\end{equation*}
In particular, hidden entropy maximization over this completion family has the unique maximizer
\begin{equation*}
\lambda^\star=\mu^\star=\frac12.
\end{equation*}
Therefore hidden entropy maximization defines a different selection problem from the visible entropy-rate selector of the present paper: it chooses one hidden completion only after introducing an additional hidden-level objective.
\end{proposition}

\begin{proof}
For the selected visible law one has $a=m$ and $b=1-m$, so every row of $K_{\lambda,\mu}$ is the same probability vector
\begin{equation*}
r_{\lambda,\mu}:=\bigl((1-m)\lambda,(1-m)(1-\lambda),m\mu,m(1-\mu)\bigr).
\end{equation*}
Hence the hidden chain is i.i.d.\ with common one-step law $r_{\lambda,\mu}$. In particular, its entropy rate equals the Shannon entropy of this single-step distribution:
\begin{equation*}
H_{\mathrm{hid}}(\lambda,\mu)
=-\sum_{x\in E}r_{\lambda,\mu}(x)\log r_{\lambda,\mu}(x).
\end{equation*}
Expanding the four terms gives
\begin{align*}
H_{\mathrm{hid}}(\lambda,\mu)
&=-(1-m)\lambda\log\bigl((1-m)\lambda\bigr)
-(1-m)(1-\lambda)\log\bigl((1-m)(1-\lambda)\bigr)\\
&\qquad -m\mu\log(m\mu)-m(1-\mu)\log\bigl(m(1-\mu)\bigr).
\end{align*}
Grouping the $\log(1-m)$, $\log m$, $\log\lambda$, $\log(1-\lambda)$, $\log\mu$, and $\log(1-\mu)$ terms yields
\begin{equation*}
H_{\mathrm{hid}}(\lambda,\mu)
=h_2(m)+(1-m)h_2(\lambda)+mh_2(\mu).
\end{equation*}
Since the binary entropy function is strictly concave on $(0,1)$ and has its unique maximizer at $1/2$, the weighted sum above is strictly concave in $(\lambda,\mu)$ and is uniquely maximized at
\begin{equation*}
\lambda^\star=\mu^\star=\frac12.
\end{equation*}
This maximizer depends on the hidden splitting coordinates $\lambda$ and $\mu$, which are invisible at the retained visible level. Therefore maximizing hidden entropy over hidden completions is a different problem from selecting the canonical visible completion by visible entropy-rate maximization.
\end{proof}

\begin{proposition}[Reduction to visible optimization]\label{prop:observable-determined-optimization-reduction}
Let
\begin{equation*}
S:\mathcal P_{\mathrm{stat}}(H^{\mathbb Z})\to \mathbb R
\end{equation*}
be an observable-determined real-valued criterion in the sense of \cref{prop:factorization-visible-law}. Let $\widetilde S$ be the unique visible-law functional satisfying
\begin{equation*}
S=\widetilde S\circ \Pi_\#.
\end{equation*}
Then, for every nonempty class
\begin{equation*}
\mathcal C\subseteq \mathcal P_{\mathrm{stat}}(H^{\mathbb Z}),
\end{equation*}
one has
\begin{equation*}
\sup_{Q\in \mathcal C} S(Q)=\sup_{\nu\in \Pi_\#\mathcal C}\widetilde S(\nu).
\end{equation*}
Moreover, a hidden law $Q^\star\in \mathcal C$ maximizes $S$ over $\mathcal C$ if and only if its visible law $\nu^\star:=\Pi_\#Q^\star$ maximizes $\widetilde S$ over $\Pi_\#\mathcal C$.

In particular, observable-determined optimization can determine at most a visible law, it cannot distinguish two hidden laws in $\mathcal C$ that lie in the same observational fiber.
\end{proposition}

\begin{proof}
By \cref{prop:factorization-visible-law}, there exists a unique map
\begin{equation*}
\widetilde S:\Pi_\#\bigl(\mathcal P_{\mathrm{stat}}(H^{\mathbb Z})\bigr)\to \mathbb R
\end{equation*}
such that
\begin{equation*}
S(Q)=\widetilde S(\Pi_\#Q)
\qquad\text{for every }Q\in \mathcal P_{\mathrm{stat}}(H^{\mathbb Z}).
\end{equation*}
Therefore
\begin{equation*}
\{S(Q):Q\in \mathcal C\}
=
\{\widetilde S(\nu):\nu\in \Pi_\#\mathcal C\}.
\end{equation*}
Taking suprema gives
\begin{equation*}
\sup_{Q\in \mathcal C} S(Q)=\sup_{\nu\in \Pi_\#\mathcal C}\widetilde S(\nu).
\end{equation*}
This proves the first assertion.

Now let $Q^\star\in \mathcal C$, and write $\nu^\star:=\Pi_\#Q^\star$. If $Q^\star$ maximizes $S$ over $\mathcal C$, then for every $\nu\in \Pi_\#\mathcal C$ there exists some $Q\in \mathcal C$ with $\Pi_\#Q=\nu$. Hence
\begin{equation*}
\widetilde S(\nu)=S(Q)\le S(Q^\star)=\widetilde S(\nu^\star),
\end{equation*}
so $\nu^\star$ maximizes $\widetilde S$ over $\Pi_\#\mathcal C$.

Conversely, if $\nu^\star$ maximizes $\widetilde S$ over $\Pi_\#\mathcal C$, then for every $Q\in \mathcal C$,
\begin{equation*}
S(Q)=\widetilde S(\Pi_\#Q)\le \widetilde S(\nu^\star)=S(Q^\star),
\end{equation*}
so $Q^\star$ maximizes $S$ over $\mathcal C$. This proves the equivalence.

Finally, if $Q_1,Q_2\in \mathcal C$ satisfy
\begin{equation*}
\Pi_\#Q_1=\Pi_\#Q_2,
\end{equation*}
then
\begin{equation*}
S(Q_1)=\widetilde S(\Pi_\#Q_1)=\widetilde S(\Pi_\#Q_2)=S(Q_2).
\end{equation*}
Thus observable-determined optimization cannot distinguish points in the same observational fiber.
\end{proof}

\begin{corollary}[Maximality of visible selection]
Within the framework of this paper, any observable-determined optimization principle can canonically determine at most a visible law. In particular, on the selected Bernoulli fiber of \cref{thm:aliased-canonical-completion}, the canonical visible selector is the finest resolution obtainable without adding hidden-side information beyond the observation experiment and retained visible observables.
\end{corollary}

\begin{proof}
The first statement is exactly the final assertion of \cref{prop:observable-determined-optimization-reduction}. The second follows by applying that proposition to the non-singleton selected fiber identified in \cref{cor:selected-visible-fiber-continuum}: because the fiber contains infinitely many hidden laws with the same visible image, no observable-determined optimization principle can distinguish among them. Therefore the visible selector is the maximal canonical resolution available at the observable level.
\end{proof}

\begin{proposition}[Hidden entropy and fiber invariance]\label{prop:hidden-entropy-not-fiber-invariant}
Fix $m\in(0,1)$ and let $\nu^\star$ be the selected visible Bernoulli$(m)$ law from \cref{thm:aliased-canonical-completion}. Then the hidden entropy rate is not constant on the observational fiber $\mathcal E_\Pi(\nu^\star)$. More precisely, for the stationary hidden laws $Q_{\lambda,\mu}$ from \cref{cor:selected-visible-fiber-continuum}, one has
\begin{equation*}
H_{\mathrm{hid}}(Q_{\lambda,\mu})
=h_2(m)+(1-m)h_2(\lambda)+mh_2(\mu),
\end{equation*}
so distinct points of the same observational fiber can have different hidden entropy rates.
\end{proposition}

\begin{proof}
By \cref{cor:selected-visible-fiber-continuum}, every law $Q_{\lambda,\mu}$ belongs to the same observational fiber $\mathcal E_\Pi(\nu^\star)$. By \cref{prop:contrast-hidden-entropy}, the hidden entropy rate of the corresponding stationary hidden chain is
\begin{equation*}
H_{\mathrm{hid}}(Q_{\lambda,\mu})
=h_2(m)+(1-m)h_2(\lambda)+mh_2(\mu).
\end{equation*}
It therefore suffices to exhibit two parameter pairs in $(0,1)^2$ for which this value is different.

Choose
\begin{equation*}
(\lambda,\mu)=\Bigl(\frac12,\frac12\Bigr)
\qquad\text{and}\qquad
(\lambda',\mu')=\Bigl(\frac14,\frac14\Bigr).
\end{equation*}
Since the binary entropy function is strictly increasing on $(0,1/2)$, one has
\begin{equation*}
h_2\Bigl(\frac14\Bigr)<h_2\Bigl(\frac12\Bigr).
\end{equation*}
Hence
\begin{align*}
H_{\mathrm{hid}}(Q_{1/4,1/4})
&=h_2(m)+(1-m)h_2\Bigl(\frac14\Bigr)+mh_2\Bigl(\frac14\Bigr)\\
&=h_2(m)+h_2\Bigl(\frac14\Bigr)\\
&<h_2(m)+h_2\Bigl(\frac12\Bigr)\\
&=H_{\mathrm{hid}}(Q_{1/2,1/2}).
\end{align*}
Thus the hidden entropy rate is not constant on $\mathcal E_\Pi(\nu^\star)$.
\end{proof}

\begin{corollary}[Fiber invariance of visible selection]
Fix $m\in(0,1)$ and the selected visible law $\nu^\star$ from \cref{thm:aliased-canonical-completion}. Then the canonical visible selector depends only on the visible feasible class and is therefore unchanged across the hidden observational fiber of $\nu^\star$, whereas hidden-entropy maximization distinguishes between hidden laws inside that same fiber.
\end{corollary}

\begin{proof}
The first statement is immediate from the definition of the selector: it is an optimizer of the visible entropy-rate functional on a visible feasible class, so once the visible feasible class is fixed, hidden realizations play no further role. In the present application, \cref{thm:aliased-canonical-completion} shows that the selected visible law is the Bernoulli$(m)$ law.

By contrast, \cref{prop:hidden-entropy-not-fiber-invariant} shows that the hidden entropy rate varies along the hidden observational fiber $\mathcal E_\Pi(\nu^\star)$. Therefore a hidden-entropy criterion distinguishes between hidden laws that are observationally indistinguishable at the selected visible level. This proves the stated contrast.
\end{proof}

\begin{proposition}[Fiber constancy of observable rules]\label{prop:observable-determined-fiber-constant}
Fix a visible stationary law $\nu$ and an observation map $\Pi$. Let $S$ be any rule defined on hidden stationary laws with the property that whenever
\begin{equation*}
\Pi_\#Q=\Pi_\#Q',
\end{equation*}
one has
\begin{equation*}
S(Q)=S(Q').
\end{equation*}
Then $S$ is constant on the observational fiber $\mathcal E_\Pi(\nu)$. In particular, any selection criterion determined only by the observation experiment and retained visible observables must be fiber-constant on each observational fiber.
\end{proposition}

\begin{proof}
Let $Q,Q'\in\mathcal E_\Pi(\nu)$. By the definition of the observational fiber in \cref{sec:observational-fibers},
\begin{equation*}
\Pi_\#Q=\nu=\Pi_\#Q'.
\end{equation*}
Therefore the defining property of $S$ gives
\begin{equation*}
S(Q)=S(Q').
\end{equation*}
Since $Q,Q'$ were arbitrary elements of $\mathcal E_\Pi(\nu)$, the rule $S$ is constant on that fiber.

For the final sentence, if a criterion is determined only by the observation experiment and retained visible observables, then any two hidden laws that induce the same retained visible specification are indistinguishable for that criterion. Hence the criterion satisfies the displayed implication above and is therefore fiber-constant.
\end{proof}

\begin{corollary}[Hidden entropy is not observable]
Fix $m\in(0,1)$ and the selected visible law $\nu^\star$ from \cref{thm:aliased-canonical-completion}. Hidden-entropy maximization on the hidden observational fiber $\mathcal E_\Pi(\nu^\star)$ is not an observable-determined criterion in the sense of \cref{prop:observable-determined-fiber-constant}.
\end{corollary}

\begin{proof}
By \cref{prop:observable-determined-fiber-constant}, every observable-determined rule must be constant on each observational fiber. But \cref{prop:hidden-entropy-not-fiber-invariant} shows that hidden entropy takes different values at different points of the single fiber $\mathcal E_\Pi(\nu^\star)$. Therefore hidden entropy cannot define an observable-determined criterion.
\end{proof}

\begin{proposition}[Hidden selectors cannot separate a fiber]\label{prop:no-observable-determined-hidden-selector}
Fix a visible stationary law $\nu$ and an observation map $\Pi$. Assume that the observational fiber $\mathcal E_\Pi(\nu)$ contains at least two distinct hidden stationary laws. Let
\begin{equation*}
T:\mathcal E_\Pi(\nu)\to \mathcal P_{\mathrm{stat}}(H^{\mathbb Z})
\end{equation*}
be a rule such that
\begin{enumerate}[label=(\roman*),leftmargin=2.5em]
\item $T(Q)\in \mathcal E_\Pi(\nu)$ for every $Q\in \mathcal E_\Pi(\nu)$,
\item $T$ is observable-determined in the sense of \cref{prop:observable-determined-fiber-constant}.
\end{enumerate}
Then $T$ is constant on $\mathcal E_\Pi(\nu)$. In particular, $T$ cannot separate distinct points of the fiber, equivalently, no observable-determined hidden selector can be injective on a non-singleton observational fiber.
\end{proposition}

\begin{proof}
By \cref{prop:observable-determined-fiber-constant}, every observable-determined rule is constant on each observational fiber. Therefore there exists some hidden stationary law $Q^\dagger\in \mathcal E_\Pi(\nu)$ such that
\begin{equation*}
T(Q)=Q^\dagger
\qquad\text{for every }Q\in \mathcal E_\Pi(\nu).
\end{equation*}
Since the fiber contains at least two distinct points, a constant map on that fiber cannot be injective. Hence $T$ cannot separate distinct points of $\mathcal E_\Pi(\nu)$.
\end{proof}

\begin{corollary}[Application to the selected Bernoulli fiber]
Fix $m\in(0,1)$ and let $\nu^\star$ be the selected visible Bernoulli$(m)$ law from \cref{thm:aliased-canonical-completion}. Then no observable-determined hidden selector can separate points of the observational fiber $\mathcal E_\Pi(\nu^\star)$. In particular, no such selector can recover a hidden law from the selected visible law in a one-to-one way.
\end{corollary}

\begin{proof}
By \cref{cor:selected-visible-fiber-continuum}, the fiber $\mathcal E_\Pi(\nu^\star)$ contains infinitely many distinct hidden stationary laws. The conclusion therefore follows immediately from \cref{prop:no-observable-determined-hidden-selector}.
\end{proof}

\begin{proposition}[Factorization through visible laws]\label{prop:factorization-visible-law}
Let
\begin{equation*}
S:\mathcal P_{\mathrm{stat}}(H^{\mathbb Z})\to \mathcal X
\end{equation*}
be a map into an arbitrary set $\mathcal X$ such that
\begin{equation*}
\Pi_\#Q=\Pi_\#Q'
\quad\Longrightarrow\quad
S(Q)=S(Q').
\end{equation*}
Then there exists a unique map
\begin{equation*}
\widetilde S:\Pi_\#\bigl(\mathcal P_{\mathrm{stat}}(H^{\mathbb Z})\bigr)\to \mathcal X
\end{equation*}
such that
\begin{equation*}
S=\widetilde S\circ \Pi_\#.
\end{equation*}
In particular, every observable-determined rule factors through the visible stationary law.
\end{proposition}

\begin{proof}
For each visible stationary law
\begin{equation*}
\nu\in \Pi_\#\bigl(\mathcal P_{\mathrm{stat}}(H^{\mathbb Z})\bigr),
\end{equation*}
choose any hidden stationary law $Q$ satisfying $\Pi_\#Q=\nu$, and define
\begin{equation*}
\widetilde S(\nu):=S(Q).
\end{equation*}
This is well-defined: if $Q'$ is another hidden stationary law with $\Pi_\#Q'=\nu$, then
\begin{equation*}
\Pi_\#Q=\nu=\Pi_\#Q',
\end{equation*}
so the hypothesis gives $S(Q)=S(Q')$.

Now let $Q\in \mathcal P_{\mathrm{stat}}(H^{\mathbb Z})$. By construction,
\begin{equation*}
\widetilde S(\Pi_\#Q)=S(Q),
\end{equation*}
so
\begin{equation*}
S=\widetilde S\circ \Pi_\#.
\end{equation*}
To prove uniqueness, suppose that another map $\widehat S$ satisfies
\begin{equation*}
S=\widehat S\circ \Pi_\#.
\end{equation*}
Let $\nu$ lie in the image of $\Pi_\#$, and choose $Q$ with $\Pi_\#Q=\nu$. Then
\begin{equation*}
\widetilde S(\nu)=S(Q)=\widehat S(\Pi_\#Q)=\widehat S(\nu).
\end{equation*}
Hence $\widetilde S=\widehat S$.
\end{proof}

\begin{remark}[Interpretive consequence]
The comparison above isolates the role of the selector. Hidden-identification problems ask whether the data determine a unique latent mechanism. Projection and lumpability problems ask when an aggregated visible process inherits a Markov structure from a hidden process. Externally parametrized maximum-entropy models begin by restricting attention to a prescribed functional family. The paper instead takes the observation experiment and retained visible observables as primitive, forms the corresponding visible feasible class, and then selects a canonical visible completion by entropy-rate maximization on that class.
\end{remark}

\section{The Binary First-Order Case}\label{sec:binary-example}

I specialize \cref{thm:fixed-marginal-iid,cor:gap-cmi} to the binary first-order case.

Let $A=\{0,1\}$ and $r=1$, and consider a stationary binary Markov chain with transition matrix
\begin{equation*}
P=
\begin{pmatrix}
1-a & a\\
b & 1-b
\end{pmatrix},
\qquad a,b\in[0,1].
\end{equation*}
Its stationary mean $m=\mathbb{P}(X_t=1)$ satisfies
\begin{equation*}
m=(1-m)a+m(1-b).
\end{equation*}
Equivalently,
\begin{equation*}
(1-m)a=m(1-b).
\end{equation*}

Suppose the retained observable is the stationary mean $m$.
Then the feasible visible class consists of all stationary binary first-order laws with one-point marginal
\begin{equation*}
\pi=(1-m,m).
\end{equation*}
By \cref{thm:fixed-marginal-iid}, the unique entropy-rate selector on this feasible class is the i.i.d.\ Bernoulli$(m)$ law.

For explicit parametrization, introduce the flow coordinate
\begin{equation*}
q:=(1-m)a=m(1-b).
\end{equation*}
Then
\begin{equation*}
a=\frac{q}{1-m},
\qquad
1-b=\frac{q}{m},
\qquad
0\le q\le \min\{m,1-m\}.
\end{equation*}
The entropy rate is therefore
\begin{equation*}
h(q)
=
(1-m)\,h_2\!\Bigl(\frac{q}{1-m}\Bigr)
+
m\,h_2\!\Bigl(\frac{q}{m}\Bigr),
\end{equation*}
where
\begin{equation*}
h_2(p):=-p\log p-(1-p)\log(1-p)
\end{equation*}
denotes binary entropy.
Because $h_2$ is strictly concave on $(0,1)$ and the map $q\mapsto (q/(1-m),q/m)$ is affine on the feasible interval, the function $q\mapsto h(q)$ is strictly concave on the feasible interval. Hence it has a unique maximizer.
By \cref{cor:binary-fixed-mean-iid}, that maximizer is attained at
\begin{equation*}
q^\star=m(1-m),
\end{equation*}
which yields
\begin{equation*}
a^\star=m,
\qquad
b^\star=1-m.
\end{equation*}
Equivalently, the selected transition matrix is
\begin{equation*}
P^\star=
\begin{pmatrix}
1-m & m\\
1-m & m
\end{pmatrix}.
\end{equation*}

The general gap functional from \cref{cor:gap-cmi} takes the form
\begin{equation*}
\Delta
=
h_2(m)-\bigl[(1-m)h_2(a)+m h_2(1-b)\bigr].
\end{equation*}
In the present binary first-order setting,
\begin{equation*}
\Delta=I(X_0,X_1)\ge 0.
\end{equation*}
Moreover,
\begin{equation*}
\Delta=0
\quad\Longleftrightarrow\quad
X_1\perp X_0
\quad\Longleftrightarrow\quad
a=m,\ b=1-m.
\end{equation*}
Thus the gap vanishes exactly at the entropy maximizer, and it is strictly positive for every same-mean comparator with residual serial dependence.

A convenient persistence coordinate is
\begin{equation*}
\rho:=1-a-b.
\end{equation*}
Under the fixed-mean constraint,
\begin{equation*}
a=m(1-\rho),
\qquad
b=(1-m)(1-\rho).
\end{equation*}
Hence $\rho=0$ is exactly the canonical point selected by entropy rate, while nonzero $\rho$ measures residual serial structure inside the same retained-mean feasible class.

The binary case also illustrates the dynamical interpretation of the selector. The selected visible law is structureless at the retained observable level, yet it may still admit nontrivial hidden implementations that remain observationally invisible.

\section{Conclusion}

I have developed an entropy-rate selection framework on feasible classes determined by retained observables in the finite-state finite-memory setting. The paper gives global characterization results for the entropy maximizer, a local geometric theory on fixed-support faces, a hidden realization result that leaves the visible law unchanged, and a local empirical consistency theorem. The aliased hidden-state example shows that a maximizing visible completion need not resolve hidden underidentification.

\end{document}